\documentclass[11pt,a4paper]{article}
\usepackage{amsfonts}
\usepackage[dvips]{graphicx}
\usepackage{amsmath}
\usepackage{mathrsfs}
\usepackage{makeidx}
\usepackage{xypic}
\usepackage[nottoc]{tocbibind}
\usepackage{pstricks}
\usepackage{bbm}
\usepackage[arrow, matrix, curve]{xy}
\usepackage{amsthm}
\usepackage{amssymb}
\usepackage{epic}
\usepackage{eepic}
\usepackage{epsfig}

\xymatrixcolsep{0.5cm}
\xymatrixrowsep{0.5cm}
\newtheorem{theorem}{Theorem}[section]
\newtheorem{proposition}{Proposition}[section]
\newtheorem{lemma}{Lemma}[section]

\newtheorem{corollary}{Corollary}[section]

\newcommand{\beqa}{\begin{eqnarray}}
\newcommand{\eeqa}{\end{eqnarray}}
\newcommand{\rf}[1]{(\ref{#1})}
\newcommand{\al}{\alpha}

\newcommand{\be}{\beta}

\newcommand{\la}{\lambda}

\newcommand{\bra}[1]{\langle\,#1\,|}
\newcommand{\ket}[1]{|\,#1\,\rangle}

\def\det{\operatorname{det}}
\numberwithin{equation}{section}

\begin{document}

\begin{flushright}
YITP-SB 13--21
\end{flushright}

\bigskip \vspace{15pt}
\begin{center}
\begin{LARGE}
\vspace*{1cm}
Complete spectrum and scalar products for the open spin-1/2
XXZ quantum chains  with non-diagonal boundary terms %with the most general boundary conditions by separation of variables 

\end{LARGE}

\vspace{50pt}

\begin{large}
{\bf S. Faldella}\footnote[1] {IMB, UMR 5584 du CNRS, Universit\'e
de Bourgogne, France, Simone.Faldella@u-bourgogne.fr},~~
{\bf N.~Kitanine}\footnote[2]{IMB, UMR 5584 du CNRS, Universit\'e
de Bourgogne, France, Nikolai.Kitanine@u-bourgogne.fr},~~
{\bf G. Niccoli}\footnote[3]
{%
YITP, Stony Brook University, New York 11794-3840, USA,
niccoli@max2.physics.sunysb.edu; etc}

\end{large}

\vspace{80pt}

\centerline{\bf Abstract} \vspace{1cm}
\parbox{12cm}{\small We use the quantum separation of variable (SOV) method to construct the eigenstates of the open XXZ chain with the most general boundary terms. The eigenstates in the inhomogeneous case are constructed in terms of solutions of a system of quadratic equations. This SOV representation permits us to compute scalar products and can be used to calculate form factors and correlation functions.}

\end{center}

\newpage

\section{Introduction}
Quantum integrable models and, in particular the spin chains, provide an extremely important tool for the non-perturbative analysis of quantum systems. The study of these models starting with the introduction of the Bethe ansatz \cite{Bet31} led to a number of important predictions for the equilibrium one-dimensional quantum systems. Recently it was shown that the quantum spin chains can be also used to describe several nontrivial physical phenomena in the out-of-equilibrium case. However the applications of the Bethe ansatz techniques to the systems out of equilibrium turned out to be much more complicated. 

One of the simplest examples of a quantum integrable model out of equilibrium is the XXZ spin chain with non-diagonal boundary terms. The study of this system  is necessary to tackle a wide range of open problems  from  the relaxation behavior of some classical stochastic processes, such as the ASEP \cite{EssD05,EssD06} (asymmetric simple exclusion processes), to the transport properties of the quantum spin systems \cite{SirPA09,Pro11}. In particular, it should lead to a better understanding of the diffusive spin transport in the spin chains. 

It turns out that most of these problems require computation of  correlation functions for the corresponding models. For the equilibrium case an efficient method of computation of  the correlation functions based on the quantum inverse scattering method (QISM) and the algebraic Bethe ansatz  \cite{FadST79} was first established for the periodic spin chains \cite{KitMT99,KitMT00} and then for the open spin chains with parallel boundary magnetic fields \cite{KitKMNST07,KitKMNST08}. Recent advances of this technique have been essential in order to study  the asymptotic behavior of the two-point dynamical correlation functions and structure factors \cite{KitKMST11,KitKMST12}. The most crucial steps of  this approach are the quantum inverse problem solution \cite{KitMT99, MaiT00} and the computation of the scalar products of the so-called on-shell and off-shell Bethe vectors \cite{Sla89}. For the quantum systems solvable by the algebraic Bethe ansatz these two steps lead to manageable expressions for the correlation functions. 

It is important to mention that the usual algebraic Bethe ansatz technique for open systems  \cite{Skl88} (based on the Cherednik reflection equation \cite{Che84}) cannot be applied directly to the spin chains with non-diagonal boundary terms  (which can be understood as non-parallel boundary magnetic fields). The first successful attempt to describe the spectrum of the XXZ spin chain with non-diagonal boundary terms was performed by Nepomechie  \cite{Nep02,Nep04} using the Baxter $T$-$Q$ equation. This method worked only for the roots of unity points and only if  the boundary terms satisfied a very particular constraint  relating the magnetic fields on the left and right boundaries. Similar constraint was obtained in \cite{CaoLSW03} within the framework of the generalized algebraic Bethe ansatz. Using a very particular gauge transformation (inspired by the Baxter \cite{Bax72,Bax72a} and Faddeev-Takhtadjan \cite{FadT79} approaches  for the XYZ spin chain) the authors constructed for the first time the eigenstates of the XXZ spin chain with non-parallel boundary magnetic fields.  It is important to mention that this approach provided a possibility to get rid of the requirement to consider only the spin chains at the roots of unity\footnote{In this case with boundary constrain alternative methods leading to Bethe equation formulations have been developed both in \cite{deGP04, NicRd05}, in the Temperley-Lieb algebraic framework, and in \cite{CraRS10, CraR12}, by combing coordinate Bethe ansatz and matrix ansatz.}.  A slightly different version of this technique, based on the vertex-IRF transformation, was proposed in \cite{YanZ07}. Even though  this method offered a more clear algebraic construction, it required one more additional constraint for the boundary parameters.  

 {\it A priori} this last method appeared to be  the most suitable for the study of correlation functions, as some crucial quantities could be computed explicitly  within this framework \cite{FilK11}. However it turns out that the computation of the scalar products of Bethe vectors remains an open  problem despite several unsuccessful attempts to solve it. Moreover, it seems that the boundary constraints remain essential to apply  the algebraic Bethe ansatz in the non-diagonal.
 
  Other approaches have been developed to deal with this spectral problem in this general setting. For the eigenvalue characterization in \cite{Gal08} a new functional method has been introduced leading to nested Bethe ansatz type equations similar to those presented previously in \cite{MurN05} in the generalized $T$-$Q$ formalism. In the so-called $q$-Onsager formalism the eigenstate construction has been addressed in \cite{BasK05a,Bas06} leading to a  characterization of the spectrum in terms of the roots of some characteristic polynomials.
  
An alternative way to construct eigenvalues and eigenvectors  for quantum integrable systems not solvable by the algebraic Bethe ansatz is the quantum separation of variables (SOV) introduced by Sklyanin \cite{Skl85,Skl92} for the quantum Toda chain. It  was  recently shown that besides the spectrum the scalar products of Bethe vectors can also be computed using the SOV technique leading to manageable expressions for the matrix elements of local operators for the cyclic sine-Gordon model, anti-periodic spin chains, SOS model and several other systems \cite{NicT10,GroMN12,GroN12,Nic13a,Nic13b,Nic13c}. 
 
 Recently this method was also applied to the open spin chains with non-diagonal boundary terms \cite{Nic12b}. This technique provided a possibility to construct the spectrum and to compute scalar products for a spin chain with non-diagonal boundaries under condition that one of boundary $K$ matrices is triangular. While this conditions  also represents a constraint its nature is quite different from  ones used for other approaches. In particular, it is a constraint only for one boundary (and not relating parameters for two boundaries). It is worth mentioning that for the XXX chain even the most general case can be reduced to this one due to the $SU(2)$ invariance of the bulk Hamiltonian. Therefore the SOV approach can be used to completely solve the XXX chain with the most general non-diagonal boundary fields. In fact, some results in this direction for this model already appear in \cite{FraSW08,FraGSW11}, where the functional version of the separation of variables of Sklyanin has been developed under general boundary conditions leading to the eigenvalues and wave-function characterizations. This functional approach, however, does not lead to the construction of the transfer matrix eigenstates in the original Hilbert space of the quantum chain, which instead can be obtained adapting to the rational 6-vertex case the SOV method developed in  \cite{Nic12b}. However, the $SU(2)$ symmetry is lost for the XXZ case and the remaining $U(1)$ symmetry is not sufficient (in general) to triangularize one of the $K$ matrices.

In this paper we study the quantum XXZ spin chains with the most general boundary terms. To construct the eigenstates and to compute the spectrum of the teansfer matrix we use two techniques: first we apply the  gauge transformation introduced in \cite{CaoLSW03} and then we apply the SOV approach. The gauge transformation provides a possibility to obtain a triangular  boundary $K$ matrix which leads to a SOV solution similar to that of \cite{Nic12b}.  These two steps  combined  lead to the description of the spectrum and the eigenstates for the most general open XXZ chain and to determinant representations for the scalar product of such SOV states. 

It is important to underline that this solution works if the  boundary constraint is not satisfied. 
More precisely,  there are two equivalent ways to construct the eigenstates. If the  constraint is satisfied one of them does not work. 

The main peculiarity of this case is the fact that the SOV analysis does not lead to a polynomial $Q$ operator and hence to the Bethe equations. Here the eigenstates are defined through a system of quadratic equations. These equations replace the Baxter $T$-$Q$ relation and permit us to completely characterize the eigenstates and the eigenvalues of the transfer matrix.

The main advantage of this approach is the simplicity of the final representation for the eigenstates in the SOV basis. On the other hand the main difficulty is the fact that this approach works only for the inhomogeneous case and the thermodynamic and homogeneous limits cannot be taken easily. In particular, it is not yet clear how to identify the ground state for the hermitian Hamiltonian and the stationary state for the model out of equilibrium. However we think that these problems can be solved within the framework of our approach.

The paper is organized as follows. In the Section 2 we describe the reflection algebra and construction of the commuting transfer matrices following \cite{Skl88}. In the Section 3 we introduce the gauge transformation and  establish the main properties of the gauged elements of the monodromy matrix. The SOV basis in terms of these operators is constructed in the Section 4. The main result of the paper is given in the Section 5 where we construct  the eigenstates of the  transfer matrix in the SOV framework. In the last section we give our expression for scalar products. The implications of these results  and some open problems are discussed in the conclusion. In the appendix we give the explicit form of the gauged transformed boundary matrices.   

\section{Reflection algebra and open spin-1/2 XXZ quantum chain}

We study in this paper the quantum open XXZ spin chain with the most general boundary terms: 
\begin{align}
H& =\sum_{i=1}^{\mathsf{N}-1}(\sigma _{i}^{x}\sigma _{i+1}^{x}+\sigma
_{i}^{y}\sigma _{i+1}^{y}+\cosh \eta \sigma _{i}^{z}\sigma _{i+1}^{z}) 
\notag \\
& +\frac{\sinh \eta }{\sinh \zeta _{-}}\left[ \sigma _{1}^{z}\cosh \zeta
_{-}+2\kappa _{-}(\sigma _{1}^{x}\cosh \tau _{-}+i\sigma _{1}^{y}\sinh \tau
_{-})\right]   \notag \\
& +\frac{\sinh \eta }{\sinh \zeta _{+}}[(\sigma _{\mathsf{N}}^{z}\cosh \zeta
_{+}+2\kappa _{+}(\sigma _{\mathsf{N}}^{x}\cosh \tau _{+}+i\sigma _{\mathsf{N%
}}^{y}\sinh \tau _{+}).  \label{H-XXZ-Non-D}
\end{align}%
This Hamiltonian acts in a tensor product ${\mathbb{C}^2}^{\otimes N}$, $\sigma_i^a$ are local spin $1/2$ operators (Pauli matrices), $\Delta =\cosh\eta$ is the anisotropy parameter and six complex boundary parameters $\zeta_\pm$, $\kappa_\pm$ and $\tau_\pm$ give the most general boundary interactions. 

 In the framework of
the quantum inverse scattering method the open XXZ spin chain is characterized by monodromy matrices $\mathcal{U}%
(\lambda )$ which are solutions of the following reflection equation:%
\begin{equation}
R_{12}(\lambda -\mu )\,\mathcal{U}_{1}(\lambda )\,R_{21}(\lambda +\mu -\eta
)\,\mathcal{U}_{2}(\mu )=\mathcal{U}_{2}(\mu )\,R_{12}(\lambda +\mu -\eta )\,%
\mathcal{U}_{1}(\lambda )\,R_{21}(\lambda -\mu ),  \label{bYB}
\end{equation}%
where the $R$-matrix is the 6-vertex
trigonometric solution of the Yang-Baxter equation:%
\begin{equation}
R_{12}(\lambda -\mu )R_{13}(\lambda )R_{23}(\mu )=R_{23}(\mu )R_{13}(\lambda
)R_{12}(\lambda -\mu ),
\end{equation}%
and
\begin{equation}
R_{12}(\lambda )= \left( 
\begin{array}{cccc}
\sinh (\lambda +\eta ) & 0 & 0 & 0 \\ 
0 & \sinh \lambda  & \sinh \eta  & 0 \\ 
0 & \sinh \eta  & \sinh \lambda  & 0 \\ 
0 & 0 & 0 & \sinh (\lambda +\eta )%
\end{array}%
\right) \in \text{End}(\mathcal{H}_{1}\otimes \mathcal{H}_{2}),
\end{equation}%
where $\mathcal{H}_{a}\simeq \mathbb{C}^{2}$ is a 2-dimensional linear space. The most general scalar solution $K(\lambda)\in\text{End}(\mathcal{H}_0\simeq \mathbb{C}^{2})$ of the reflection equation is the
following $2\times 2$ matrix:%
\begin{equation}
K(\lambda ;\zeta ,\kappa ,\tau )=\frac{1}{\sinh \zeta }\left( 
\begin{array}{cc}
\sinh (\lambda -\eta /2+\zeta ) & \kappa e^{\tau }\sinh (2\lambda -\eta ) \\ 
\kappa e^{-\tau }\sinh (2\lambda -\eta ) & \sinh (\zeta -\lambda +\eta /2)%
\end{array}%
\right) ,  \label{ADMFKK}
\end{equation}%
where $\zeta ,$ $\kappa $ and $\tau $ are arbitrary complex parameters.

 Starting from the scalar $K$-matrix  following Sklyanin \cite{Skl88} we can construct new solutions in 
the 2$^{\mathsf{N}}$-dimensional representation space:%
\begin{equation}
\mathcal{H}= \otimes _{n=1}^{\mathsf{N}}\mathcal{H}_{n}.
\end{equation}%
More precisely it is possible to construct two classes of solutions to the same reflection equation (\ref{bYB}). First we
define:%
\begin{equation}
K_{-}(\lambda )=K(\lambda ;\zeta _{-},\kappa _{-},\tau _{-}),\text{ \ \ \ \ }%
K_{+}(\lambda )=K(\lambda +\eta ;\zeta _{+},\kappa _{+},\tau _{+}),
\end{equation}%
where $\zeta _{\pm },\kappa _{\pm },\tau _{\pm }$ are arbitrary complex
parameters. Then we construct  the (bulk) inhomogeneous  ``left to right" monodromy matrix
\begin{equation}
M_{0}(\lambda ) =R_{0\mathsf{N}}(\lambda -\xi _{\mathsf{N}}-\eta /2)\dots
R_{01}(\lambda -\xi _{1}-\eta /2)=\left( 
\begin{array}{cc}
A(\lambda ) & B(\lambda ) \\ 
C(\lambda ) & D(\lambda )\end{array}\right).\label{T} \end{equation}
In a similar way we can construct the ``right to left" monodromy matrix
\begin{equation}
\hat{M}(\lambda ) =(-1)^{\mathsf{N}}\,\sigma _{0}^{y}\,M^{t_{0}}(-\lambda
)\,\sigma _{0}^{y}.%
 \label{Mhat}
\end{equation}%
The inhomogeneity parameters $\xi_j$ are arbitrary complex numbers, in this paper we need to keep them  generic, the physical case corresponds to the homogeneous limit\footnote{%
Here the homogeneous
limit corresponds to $\xi _{m}=0$ for $m=1,\ldots ,\mathsf{N}$.}. 
$M_{0}(\lambda )\in $ End$(\mathcal{H}_{0}\otimes \mathcal{H})$,
satisfies the Yang-Baxter relation%
\begin{equation}
R_{12}(\lambda -\mu )M_{1}(\lambda )M_{2}(\mu )=M_{2}(\mu )M_{1}(\lambda
)R_{12}(\lambda -\mu ).  \label{YB}
\end{equation}%
Then we can define the boundary monodromy matrices $\mathcal{U}_{\pm }(\lambda
)\in $ End$(\mathcal{H}_{0}\otimes \mathcal{H})$ as  follows:%
\begin{eqnarray}
\mathcal{U}_{-}(\lambda ) &=&M_{0}(\lambda )K_{-}(\lambda )\hat{M}%
_{0}(\lambda )=\left( 
\begin{array}{cc}
\mathcal{A}_{-}(\lambda ) & \mathcal{B}_{-}(\lambda ) \\ 
\mathcal{C}_{-}(\lambda ) & \mathcal{D}_{-}(\lambda )%
\end{array}%
\right) , \\
\mathcal{U}_{+}^{t_{0}}(\lambda ) &=&M_{0}^{t_{0}}(\lambda
)K_{+}^{t_{0}}(\lambda )\hat{M}_{0}^{t_{0}}(\lambda )=\left( 
\begin{array}{cc}
\mathcal{A}_{+}(\lambda ) & \mathcal{C}_{+}(\lambda ) \\ 
\mathcal{B}_{+}(\lambda ) & \mathcal{D}_{+}(\lambda )%
\end{array}%
\right) .
\end{eqnarray}%
$\mathcal{U}_{-}(\lambda )$ and $\mathcal{V}_{+}(\lambda )= \mathcal{U}%
_{+}^{t_{0}}(-\lambda )$ are two classes of solutions of the reflection
equation (\ref{bYB}). It is shown by Sklyanin \cite{Skl88} that
from this couple of monodromy matrices one can define the following
commuting family of transfer matrices:%
\begin{equation}
\mathcal{T}(\lambda )= \text{tr}_{0}\{K_{+}(\lambda )\,M(\lambda
)\,K_{-}(\lambda )\hat{M}(\lambda )\}=\text{tr}_{0}\{K_{+}(\lambda )\mathcal{%
U}_{-}(\lambda )\}=\text{tr}_{0}\{K_{-}(\lambda )\mathcal{U}_{+}(\lambda
)\}\in \text{\thinspace End}(\mathcal{H}).  \label{transfer}
\end{equation}%
In this paper we characterize the complete spectrum (eigenvalue $\&$
eigenstates) of this transfer matrix %and compute the matrix elements of the
%identity in the transfer matrix eigenstates 
for the most general class of
non-diagonal boundary $K$-matrices in this way generalizing the results of 
\cite{Nic12b}. Our analysis applies also to the open spin-1/2 XXZ quantum chain
with the most general non-diagonal boundary terms (\ref{H-XXZ-Non-D}), as this
Hamiltonian is obtained in the homogeneous limit by the following derivative
of the transfer matrix (\ref{transfer}):%
\begin{equation}
H=\frac{2(\sinh \eta )^{1-2\mathsf{N}}}{\text{tr}\{K_{+}(\eta /2)\}\,\text{tr%
}\{K_{-}(\eta /2)\}}\frac{d}{d\lambda }\mathcal{T}(\lambda )_{\,\vrule %
height13ptdepth1pt\>{\lambda =\eta /2}\!}+\text{constant.}  \label{Ht}
\end{equation}
 If the boundary $K$ matrices are diagonal ($\kappa_\pm=0$) the eigenstates of the transfer matrix can be constructed using the algebraic Bethe ansatz \cite{Skl88}. In the non-diagonal case it turns out to be impossible as the ferromagnetic state (for example with all the spins up) is no more the highest weight vector for the reflection algebra. Several methods were applied to overcome this difficulty. In particular, in  \cite{CaoLSW03} the authors proposed a gauge transformation to diagonalize one of the $K$ matrices and to make the second one triangular. Such gauge transformation exists only if the boundary parameters satisfy a {\it boundary constraint} \cite{Nep02}.  Our goal here is to construct the eigenstates for the most general values of the boundary parameters. 

\subsection{First fundamental properties}

%We will need for the future analysis some properties of the generators of the reflection algebra

 Here we establish some properties of the generators $\mathcal{A}_{-}(\lambda ),$ $%
\mathcal{B}_{-}(\lambda ),$ $\mathcal{C}_{-}(\lambda )$ and $\mathcal{D}_{-}(\lambda )$ of the reflection algebra %
 which play a
fundamental role in the solution of the transfer matrix $\mathcal{T}(\lambda
)$ spectral problem. 

First of all their commutation relations follow from the reflection equation (\ref{bYB}). Using these relations it was shown by Sklyanin that the {\it quantum determinant} 
\begin{align}
\frac{\det_{q}\mathcal{U}_{-}(\lambda )}{\sinh (2\lambda -2\eta )}& = 
\mathcal{A}_{-}(\epsilon \lambda +\eta /2)\mathcal{A}_{-}(\eta /2-\epsilon
\lambda )+\mathcal{B}_{-}(\epsilon \lambda +\eta /2)\mathcal{C}_{-}(\eta
/2-\epsilon \lambda )  \label{q-detU_1} \\
& =\mathcal{D}_{-}(\epsilon \lambda +\eta /2)\mathcal{D}_{-}(\eta
/2-\epsilon \lambda )+\mathcal{C}_{-}(\epsilon \lambda +\eta /2)\mathcal{B}%
_{-}(\eta /2-\epsilon \lambda ),  \label{q-detU_2}
\end{align}%
where $\epsilon =\pm 1$,
is a central element of the reflection algebra
\begin{equation}
\lbrack \det_{q}\mathcal{U}_{-}(\lambda ),\mathcal{U}_{-}(\mu )]=0.
\end{equation}
The quantum determinant  admits the following explicit expressions:%
\begin{eqnarray}
\det_{q}\mathcal{U}_{-}(\lambda ) &=&\det_{q}K_{-}(\lambda
)\det_{q}M_{0}(\lambda )\det_{q}M_{0}(-\lambda )  \label{q-detU_-exp} \\
&=&\sinh (2\lambda -2\eta )\mathsf{A}_{-}(\lambda +\eta /2)\mathsf{A}%
_{-}(-\lambda +\eta /2),
\end{eqnarray}%
where \begin{equation}
%\zeta _{2\mathsf{N}+1}^{(0)}=\zeta _{2\mathsf{N}+1}^{(1)}=\eta /2,\text{ \ \
%\ \ }\zeta _{2\mathsf{N}+2}^{(0)}=\zeta _{2\mathsf{N}+2}^{(1)}=\eta /2+i\pi
%/2,\text{ \ \ }
\det_{q}M(\lambda )=a(\lambda +\eta /2)d(\lambda -\eta /2),
\label{bulk-q-det}\end{equation} is the bulk quantum determinant and%
\begin{equation}
\det_{q}K_{\pm}(\lambda )=\mp\sinh (2\lambda \pm2\eta )g_{\pm}(\lambda +\eta
/2)g_{\pm}(-\lambda +\eta /2).
\end{equation}%
 We have used here the following notations:%
\begin{equation}\label{eigenA}
\mathsf{A}_{-}(\lambda )= g_{-}(\lambda )a(\lambda )d(-\lambda ),\text{
\ }d(\lambda )= a(\lambda -\eta ),\text{ \ \ }a(\lambda )=
\prod_{n=1}^{\mathsf{N}}\sinh (\lambda -\xi _{n}+\eta /2),
\end{equation}%
\begin{equation}
g_{\pm }(\lambda )= \frac{\sinh (\lambda +\alpha _{\pm }- \eta
/2)\cosh (\lambda +\beta _{\pm }- \eta /2)}{\sinh \alpha _{\pm }\cosh
\beta _{\pm }},  \label{g_PM}
\end{equation}%
where $\alpha _{\pm }$\ and $\beta _{\pm }$ are defined in terms of the
boundary parameters by:%
\begin{equation}
\sinh \alpha _{\pm }\cosh \beta _{\pm }= \frac{\sinh \zeta _{\pm }}{%
2\kappa _{\pm }},\text{ \ \ \ \ \ }\cosh \alpha _{\pm }\sinh \beta _{\pm
}= \frac{\cosh \zeta _{\pm }}{2\kappa _{\pm }}.  \label{alfa-beta}
\end{equation}

\begin{proposition}
The inverse monodromy matrix can be expressed in terms of the quantum determinant as  follows
\begin{equation}
\mathcal{U}_{-}^{-1}(\lambda +\eta /2)=\frac{\sinh (2\lambda -2\eta )}{%
\det_{q}\mathcal{U}_{-}(\lambda )}\mathcal{U}_{-}(\eta /2-\lambda ).
\label{Inverse}
\end{equation}
\end{proposition}
\begin{proof}
We  first observe   that the following
identity holds
\begin{equation}
K_{-}^{-1}(\lambda +\eta /2)=\frac{\sinh(2\lambda -2\eta )}{\det_{q}K_{-}(%
\lambda )}K_{-}(\eta /2-\lambda ),  \label{K-inverse}
\end{equation}%
then the identity (\ref{Inverse}) follows by computing the matrix products $\mathcal{U}_{-}(\eta /2+\lambda )\mathcal{U}_{-}(\eta /2-\lambda )$ using \rf{q-detU_-exp}, \rf{K-inverse} and the following identities:
\begin{equation}
\hat{M}(\pm \lambda +\eta /2)=(-1)^{\mathsf{N}}\det_{q}M_{0}(\mp \lambda )M^{-1}(\mp \lambda +\eta /2).
\label{M-inverse}
\end{equation}\end{proof}We will also use the following properties of the generators 

\begin{proposition}[Prop. 2.1 of \protect\cite{Nic12b}]\label{la-la}
The generator families $\mathcal{A}_{-}(\lambda )$ and $\mathcal{D}%
_{-}(\lambda )$ are related by the following parity relation:%
\begin{eqnarray}
\mathcal{A}_{-}(\lambda ) &=&\frac{\sinh (2\lambda -\eta )}{\sinh 2\lambda }%
\mathcal{D}_{-}(-\lambda )+\frac{\sinh \eta }{\sinh 2\lambda }\mathcal{D}%
_{-}(\lambda ),  \label{Sym-A-D-} \\
\mathcal{D}_{-}(\lambda ) &=&\frac{\sinh (2\lambda -\eta )}{\sinh 2\lambda }%
\mathcal{A}_{-}(-\lambda )+\frac{\sinh \eta }{\sinh 2\lambda }\mathcal{A}%
_{-}(\lambda ),
\end{eqnarray}%
while for the other two families the following parity relations hold:%
\begin{equation}
\mathcal{B}_{-}(-\lambda )=-\frac{\sinh (2\lambda +\eta )}{\sinh (2\lambda
-\eta )}\mathcal{B}_{-}(\lambda )\text{ },\text{ \ }\mathcal{C}_{-}(-\lambda
)=-\frac{\sinh (2\lambda +\eta )}{\sinh (2\lambda -\eta )}\mathcal{C}%
_{-}(\lambda ).  \label{Sym-B-C-}
\end{equation}%
\end{proposition}
It is important to mention  that similar statements hold for the reflection
algebra generated by $\mathcal{U}_{+}(\lambda )$. In fact, they are simply
consequences of the previous proposition  taking into account that $%
\mathcal{U}_{+}^{t_0}(-\lambda )$ satisfies the same reflection equation of $%
\mathcal{U}_{-}(\lambda )$.

%Moreover, it holds:
For some particular choices of boundary parameters the transfer matrix is hermitian, more precisely:
\begin{proposition}[Prop. 2.3 of \protect\cite{Nic12b}]
The monodromy matrix \thinspace $\mathcal{U}_{\pm }(\lambda )$ satisfy the
following transformation properties under Hermitian conjugation:\newline
\textsf{I)} Under the condition $\eta \in i\mathbb{R}$ (massless regime), it
holds: 
\begin{equation}
\mathcal{U}_{\pm }(\lambda )^{\dagger }=\left[ \mathcal{U}_{\pm }(-\lambda
^{\ast })\right] ^{t_{0}},  \label{ml-Hermitian_U}
\end{equation}%
\ \ \ \ \ for $\{i\tau _{\pm },i\kappa _{\pm },i\zeta _{\pm },\xi
_{1},...,\xi _{\mathsf{N}}\}\in \mathbb{R}^{\mathsf{N}+3}.$\newline
\textsf{II)} Under the condition $\eta \in \mathbb{R}$ (massive regime), it
holds: 
\begin{equation}
\mathcal{U}_{\pm }(\lambda )^{\dagger }=\left[ \mathcal{U}_{\pm }(\lambda
^{\ast })\right] ^{t_{0}},  \label{m-Hermitian_U}
\end{equation}%
\ \ \ \ \ for $\{\tau _{\pm },\kappa _{\pm },\zeta _{\pm },i\xi
_{1},...,i\xi _{\mathsf{N}}\}\in \mathbb{R}^{\mathsf{N}+3}.$\newline
Under the same conditions on the parameters of the representation it holds: 
\begin{equation}
\mathcal{T}(\lambda )^{\dagger }=\mathcal{T}(\lambda ^{\ast }),
\label{I-Hermitian_T}
\end{equation}%
i.e. $\mathcal{T}(\lambda )$ defines a one-parameter\ family of normal
operators which are self-adjoint both for $\lambda $ real and imaginary.
\end{proposition}

\begin{proposition}
The  transfer matrix $\mathcal{T}(\lambda )$ is even in the
spectral parameter $\lambda $:%
\begin{equation}
\mathcal{T}(-\lambda )=\mathcal{T}(\lambda ),  \label{even-transfer}
\end{equation}%
\end{proposition}

These properties were used in \cite{Nic12b} to construct the quantum separated variables for the boundary XXZ chain. This method can be used directly if one of the $K$ matrices is triangular. To go beyond this constraint we use the gauge transformation introduced in \cite{CaoLSW03}.

\section{Gauge transformations and essential properties}

\subsection{Definitions}

For arbitrary complex parameters $\alpha$ and $\beta$  we introduce the following two matrices   %
\begin{equation}
\bar{G}(\lambda  |\beta ) = (X(\lambda  |\beta
),Y(\lambda  |\beta )),\text{ \ \ }\tilde{G}(\lambda 
|\beta )= (X(\lambda |\beta+1 ),Y(\lambda
|\beta-1 )) 
\end{equation}
where we have defined the following columns%
\begin{equation}
X(\lambda  |\beta ) = \left( 
\begin{array}{c}
e^{-\left[ \lambda +( \alpha+\beta )\eta \right] } \\ 
1%
\end{array}%
\right) ,\text{ \ \ \ \ \ \ \ }Y(\lambda  |\beta )= \left( 
\begin{array}{c}
e^{-\left[ \lambda +(\alpha-\beta )\eta \right] } \\ 
1%
\end{array}%
\right) .  \label{X,Y-definitions} \end{equation}
Evidently these matrices depend also on $\alpha $ but as this parameter will not vary in the following computations we omit this argument 
 for simplicity.
It is not difficult to compute the inverse matrices:%
\begin{equation}
\bar{G}^{-1}(\lambda  |\beta  )=  \left( 
\begin{array}{c}
\bar{Y}(\lambda  |\beta  ) \\ 
\bar{X}(\lambda  |\beta )%
\end{array}%
\right) ,\text{ \ \ \ \ \ \ \ \ \ \ \ }\tilde{G}^{-1}(\lambda)= \left( 
\begin{array}{c}
\tilde{Y}(\lambda  |\beta-1 ) \\ 
\tilde{X}(\lambda  |\beta+1 )%
\end{array}%
\right)
\end{equation}%
in terms of the following rows
\begin{align}
\bar{X}(\lambda  |\beta)& = \frac{ e^{\left( \lambda
+\alpha \eta \right)}}{2 \sinh \beta \eta }\left( 1,-e^{-\left[
\lambda +(\alpha +\beta )\eta \right] }\right),  \nonumber\\
\bar{Y}(\lambda   |\beta )& = \frac{ e^{\left( \lambda
+\alpha \eta \right) }}{2\sinh\beta \eta }\left( -1,e^{-\left[
\lambda +(\alpha  -\beta )\eta \right] }\right)  \label{XY-bar}
\end{align}
\begin{equation}
\tilde{X}(\lambda  |\beta
 )= e^{\eta }\frac{\sinh \beta \eta }{\sinh (\beta
-1)\eta }\bar{X}(\lambda  |\beta ),\text{ \ }\tilde{Y}(\lambda
 )= e^{\eta }\frac{\sinh \beta \eta }{\sinh (\beta
+1)\eta }\bar{Y}(\lambda  |\beta  ), \label{XY-tilde}
\end{equation}

\subsection{Gauge transformed bulk and boundary operators}

%\subsubsection{Preliminary definitions and properties}
The gauge transformation now can be applied to the local $R$-matrices. We apply it to every $R$ matrix in the {\it auxiliary space}
\begin{equation}
R_{0a}(\lambda -\xi _{a}-\eta /2|\beta)=\tilde{G}^{-1}(\lambda -\eta /2  |\beta+\mathsf{N}-a)R_{0a}(\lambda -\xi _{a}-\eta /2){G}(\lambda -\eta /2| \beta+\mathsf{N}-a+1)
\end{equation}
Now we can construct the gauge transformed  bulk monodromy matrix. Taking the
 product of $R$ matrices as in (\ref{T})%
\begin{equation}
M(\lambda |\beta) = \tilde{G}^{-1}(\lambda -\eta /2| \beta)\,M(\lambda )%
\tilde{G}(\lambda -\eta /2| \beta+\mathsf{N}) \\
= \left( 
\begin{array}{ll}
A(\lambda |\beta) & B(\lambda |\beta) \\ 
C(\lambda |\beta) & D(\lambda |\beta)%
\end{array}%
\right) , \end{equation}
where all the new monodromy matrix elements can be easily expressed in terms of the initial monodromy matrix and the row and columns defined above, for example
\begin{equation}
 B(\lambda |\beta)= \tilde{Y}(\lambda
-\eta /2| \beta-1)M(\lambda )Y(\lambda -\eta /2| \beta+\mathsf{N}-1).
\end{equation}
In a similar way we can apply the second gauge transformation to the ``right to left" monodromy matrix (\ref{Mhat})
\begin{equation}
\hat{M}(\lambda |\beta) = \bar{G}^{-1}(\eta /2-\lambda| \beta+\mathsf{N} )\,\hat{M}%
(\lambda )\bar{G}(\eta /2-\lambda| \beta ) \\
= \left( 
\begin{array}{ll}
\bar{A}(\lambda |\beta) & \bar{B}(\lambda |\beta) \\ 
\bar{C}(\lambda |\beta) & \bar{D}(\lambda |\beta)%
\end{array}%
\right).
\end{equation}%
We can define corresponding two-row monodromy matrix
\begin{equation}
\mathsf{U}_{-}(\lambda |\beta) =  \tilde{G}^{-1}(\lambda -\eta /2| \beta)\,\mathcal{U}_{-}(\lambda )\,\tilde{G}(\eta /2-\lambda| \beta ) =\left( 
\begin{array}{ll}
\widehat{\mathcal{A}}_{-}(\lambda |\beta+2) & \widehat{\mathcal{B}}_{-}(\lambda
|\beta) \\ 
\,\widehat{\mathcal{C}}_{-}(\lambda |\beta+2) & \widehat{\mathcal{D}}_{-}(\lambda
|\beta)%
\end{array}%
\right) .\label{U-gauge}\end{equation}
Note that this definition leads to a non-trivial ``dynamical" boundary bulk decomposition:
\begin{align}
\left(\begin{array}{l}
\widehat{\mathcal{A}}_{-}(\lambda |\beta+2) \\ 
\,\widehat{\mathcal{C}}_{-}(\lambda |\beta+2)%
\end{array}\right)= &M(\lambda |\beta)\bar{K}_-(\lambda|\beta)\left( 
\begin{array}{l}
\bar{A}(\lambda |\beta+1)  \\ 
\bar{C}(\lambda |\beta+1) %
\end{array}%
\right)\label{boundary-bulkAC}\\
\left(\begin{array}{l}
 \widehat{\mathcal{B}}_{-}(\lambda
|\beta) \\ 
 \widehat{\mathcal{D}}_{-}(\lambda
|\beta)%
\end{array}%
\right)=&M(\lambda |\beta)K_-(\lambda|\beta)\left( 
\begin{array}{l}
\bar{B}(\lambda |\beta-1)  \\ 
\bar{D}(\lambda |\beta-1) %
\end{array}%
\right),\label{boundary-bulkBD}
\end{align}
where
\begin{align}
K_-(\lambda|\beta)=&\tilde{G}^{-1}(\lambda -\eta /2| \beta+\mathsf{N})\,K_-(\lambda) \, \bar{G}(\eta /2-\lambda| \beta+\mathsf{N}-1 ), \label{K-beta}\\ 
\bar{K}_-(\lambda|\beta)=&\tilde{G}^{-1}(\lambda -\eta /2| \beta+\mathsf{N})\,K_-(\lambda) \, \bar{G}(\eta /2-\lambda| \beta+\mathsf{N}+1 ). \label{barK-beta}
\end{align}
It is more convenient to normalize the new double row monodromy matrix in the following way
%
%Then, defined the following rescaled boundary gauged operators:%
\begin{equation}
\mathcal{U}_{-}(\lambda |\beta)\equiv e^{-\lambda +\eta /2}\mathsf{U}_{-}(\lambda |\beta)=\left(\begin{array}{ll}
\mathcal{A}_{-}(\lambda |\beta+2) & \mathcal{B}_{-}(\lambda
|\beta) \\ 
\,\mathcal{C}_{-}(\lambda |\beta+2) & \mathcal{D}_{-}(\lambda
|\beta)%
\end{array}\right).
%\mathcal{A}_{-}(\lambda |\beta) &=&e^{-\lambda +\eta /2}\widehat{\mathcal{A}}%
%_{-}(\lambda |\beta),\text{ \ }\mathcal{B}_{-}(\lambda |\beta)=e^{-\lambda +\eta /2}%
%\widehat{\mathcal{B}}_{-}(\lambda |\beta), \\
%\mathcal{C}_{-}(\lambda |\beta) &=&e^{-\lambda +\eta /2}\widehat{\mathcal{C}}%
%_{-}(\lambda |\beta),\text{ \ }\mathcal{D}_{-}(\lambda |\beta)=e^{-\lambda +\eta /2}%
%\widehat{\mathcal{D}}_{-}(\lambda |\beta),
\end{equation}%

\subsection{Properties of the gauge transformed operators }

The commutation relation of the generators of the reflection algebra are given by the equation 
(\ref{bYB}). Applying the gauge transformation one can derive the dynamical commutation relation for the transformed generators:

\begin{lemma}
The following commutations relations hold for the gauged transformed
reflection algebra generators:%
\begin{equation}
\mathcal{B}_{-}(\lambda _{2}|\beta)\mathcal{B}_{-}(\lambda _{1}|\beta-2)=\mathcal{B}%
_{-}(\lambda _{1}|\beta)\mathcal{B}_{-}(\lambda _{2}|\beta-2),  \label{CRM-BB}
\end{equation}%
\begin{eqnarray}
\mathcal{A}_{-}(\lambda _{2}|\beta+2)\mathcal{B}_{-}(\lambda _{1}|\beta) &=&\frac{%
\sinh (\lambda _{1}-\lambda _{2}+\eta )\sinh (\lambda _{2}+\lambda _{1}-\eta
)}{\sinh (\lambda _{1}-\lambda _{2})\sinh (\lambda _{1}+\lambda _{2})}%
\mathcal{B}_{-}(\lambda _{1}|\beta)\mathcal{A}_{-}(\lambda _{2}|\beta)  \notag \\
&+&\frac{\sinh (\lambda _{1}+\lambda _{2}-\eta )\sinh (\lambda _{1}-\lambda
_{2}+(\beta -1)\eta )\sinh \eta }{\sinh (\lambda _{2}-\lambda _{1})\sinh
(\lambda _{1}+\lambda _{2})\sinh (\beta -1)\eta )}\mathcal{B}_{-}(\lambda
_{2}|\beta)\mathcal{A}_{-}(\lambda _{1}|\beta)  \notag \\
&+&\frac{\sinh \eta \sinh (\lambda _{1}+\lambda _{2}-\beta \eta )}{\sinh
(\lambda _{1}+\lambda _{2})\sinh (\beta -1)\eta }\mathcal{B}_{-}(\lambda
_{2}|\beta)\mathcal{D}_{-}(\lambda _{1}|\beta),  \label{CMR-AB-Left}
\end{eqnarray}%
\begin{align}
\mathcal{B}_{-}(\lambda _{1}|\beta)\mathcal{D}_{-}(\lambda _{2}|\beta)& =\frac{\sinh
(\lambda _{1}-\lambda _{2}+\eta )\sinh (\lambda _{2}+\lambda _{1}-\eta )}{%
\sinh (\lambda _{1}-\lambda _{2})\sinh (\lambda _{1}+\lambda _{2})}\mathcal{D%
}_{-}(\lambda _{2}|\beta+2)\mathcal{B}_{-}(\lambda _{1}|\beta)  \notag \\
& -\frac{\sinh (\lambda _{2}-\lambda _{1}+(\beta+1 )\eta )\sinh (\lambda
_{2}+\lambda _{1}-\eta )}{\sinh (\lambda _{1}-\lambda _{2})\sinh (\lambda
_{2}+\lambda _{1})\sinh (\beta+1 )\eta }\mathcal{D}_{-}(\lambda _{1}|\beta+2)%
\mathcal{B}_{-}(\lambda _{2}|\beta)  \notag \\
& -\frac{\sinh \eta \sinh (\lambda _{2}+\lambda _{1}+\beta \eta )}{\sinh
(\lambda _{2}+\lambda _{1})\sinh (\beta+1 )\eta }\mathcal{A}_{-}(\lambda
_{1}|\beta+2)\mathcal{B}_{-}(\lambda _{2}|\beta),  \label{BD-DB-CMR}
\end{align}%
\begin{align}
& \mathcal{A}_{-}(\lambda _{1}|\beta+2)\mathcal{A}_{-}(\lambda _{2}|\beta+2)-\frac{%
\sinh \eta \sinh (\lambda _{1}+\lambda _{2}-(\beta )\eta )}{\sinh (\lambda
_{1}+\lambda _{2})\sinh (\beta -1)\eta }\mathcal{B}_{-}(\lambda _{1}|\beta)%
\mathcal{C}_{-}(\lambda _{2}|\beta+2)\left. =\right.  \notag \\
& \text{ \ \ \ \ \ \ \ \ \ \ \ \ \ \ \ \ \ \ \ \ }\mathcal{A}%
_{-}(\lambda _{2}|\beta+2)\mathcal{A}_{-}(\lambda _{1}|\beta+2)-\frac{\sinh \eta
\sinh (\lambda _{1}+\lambda _{2}-\beta \eta )}{\sinh (\lambda
_{1}+\lambda _{2})\sinh (\beta -1)\eta }\mathcal{B}_{-}(\lambda _{2}|\beta)%
\mathcal{C}_{-}(\lambda _{1}|\beta+2).  \label{CMR-AA-BC}
\end{align}
\end{lemma}
It is not the complete list of relations that one can get from the reflection equation but it is all we need to construct the SOV representations. They can be seen as the the commutation relations of the generators of the  dynamical reflection algebra \cite{FilK11}.  %as the gauge transformation in the auxiliary space produces the vertex-IRF transformation for the $R$-matrices in the reflection equation (\ref{bYB}).

%The first two commutation relations appear first in the paper \cite{CaoLSW03},
%the proof of  other two  relations is given in the Appendix B

It is also possible to establish symmetry properties similar to the Proposition \ref{la-la}

\begin{proposition}
The generators  $\mathcal{A}_{-}(\lambda |\beta)$ and $\mathcal{D}%
_{-}(\lambda |\beta)$ are related by the following parity relation:%
\begin{align}
\mathcal{A}_{-}(\lambda |\beta)& =-\frac{\sinh \eta \sinh (2\lambda -(\beta
-1)\eta )}{\sinh 2\lambda \sinh (\beta -2)\eta }\mathcal{D}_{-}(\lambda
|\beta)+\frac{\sinh (2\lambda -\eta )\sinh (\beta -1)\eta }{\sinh 2\lambda
\sinh (\beta -2)\eta }\mathcal{D}_{-}(-\lambda |\beta),  \label{parity-m-1} \\
\mathcal{D}_{-}(\lambda |\beta)& =\frac{\sinh \eta \sinh (2\lambda +(\beta
-1)\eta )}{\sinh 2\lambda \sinh \beta \eta }\mathcal{A}_{-}(\lambda |\beta)+%
\frac{\sinh (2\lambda -\eta )\sinh (\beta -1)\eta }{\sinh 2\lambda \sinh
\beta \eta }\mathcal{A}_{-}(-\lambda |\beta),  \label{parity-m-3}
\end{align}%
while for the other two generators the following parity relations hold:%
\begin{equation}
\mathcal{B}_{-}(-\lambda |\beta)=-\frac{\sinh (2\lambda +\eta )}{\sinh (2\lambda
-\eta )}\mathcal{B}_{-}(\lambda |\beta)\text{ },\text{ \ }\mathcal{C}%
_{-}(-\lambda |\beta)=-\frac{\sinh (2\lambda +\eta )}{\sinh (2\lambda -\eta )}%
\mathcal{C}_{-}(\lambda |\beta).  \label{parity-m-2}
\end{equation}%
\end{proposition}
These relations  can be obtained from Proposition \ref{la-la} by direct computation.

\begin{proposition}
The inverse transformed double-row monodromy matrix can be written in terms of the  quantum determinant of the reflection algebra
\begin{equation}
\mathcal{U}_{-}^{-1}(\lambda +\eta /2|\beta)=\frac{\sinh (2\lambda -2\eta
)}{\det_{q}\mathcal{U}_{-}(\lambda )}\mathcal{U}_{-}(\eta /2-\lambda |\beta),
\label{Inversion-formula}
\end{equation}%
where the following representation holds for the quantum determinant, for
both $\epsilon =\pm 1$:%
\begin{align}
\frac{\det_{q}\mathcal{U}_{-}(\lambda )}{\sinh (2\lambda -2\eta )}& =%
\mathcal{A}_{-}(\epsilon \lambda +\eta /2|\beta+2)\mathcal{A}_{-}(\eta
/2-\epsilon \lambda |\beta+2)+\mathcal{B}_{-}(\epsilon \lambda +\eta /2|\beta)%
\mathcal{C}_{-}(\eta /2-\epsilon \lambda |\beta+2)  \label{gauge-q-det-A} \\
& =\mathcal{D}_{-}(\epsilon \lambda +\eta /2|\beta)\mathcal{D}_{-}(\eta
/2-\epsilon \lambda |\beta)+\mathcal{C}_{-}(\epsilon \lambda +\eta /2|\beta+2)%
\mathcal{B}_{-}(\eta /2-\epsilon \lambda |\beta).  \label{gauge-q-det-D}
\end{align}
\end{proposition}

\begin{proof}
 Using the definition of the gauge transformation it is easy to see that it holds:%
\begin{equation}
\mathcal{U}_{-}(\lambda +\eta /2|\beta)= e^{-\lambda}\tilde{G}^{-1}(\lambda|\beta )%
\mathcal{U}_{-}(\lambda +\eta /2)\tilde{G}(-\lambda|\beta ),
\end{equation}%
 then we obtain:%
\begin{eqnarray}
\mathcal{U}_{-}(\lambda +\eta /2|\beta)\,\mathcal{U}_{-}( \eta /2-\lambda|\beta)
&=& \tilde{G}^{-1}(\lambda|\beta )\,\mathcal{U}%
_{-}(\lambda +\eta /2)\,\mathcal{U}_{-}(\eta /2-\lambda )\,\tilde{G}(\lambda|\beta )  \notag \\
&=&\frac{\det_{q}\mathcal{U}_{-}(\lambda )}{\sinh (2\lambda -2\eta )},
\end{eqnarray}%
and similarly:%
\begin{equation}
\mathcal{U}_{-}( \eta /2-\lambda|\beta)\,\mathcal{U}_{-}(\lambda +\eta /2|\beta)
=\frac{\det_{q}\mathcal{U}_{-}(\lambda )}{\sinh (2\lambda -2\eta )}.
\end{equation}
%\begin{eqnarray}
%\mathsf{\tilde{U}}_{-}(\lambda -\eta /2|\beta)\mathsf{U}_{-}(\lambda +\eta /2|\beta)
%&=&\tilde{G}_{m}^{-1}(-\lambda )\widetilde{\mathcal{U}}_{-}(\lambda -\eta /2)%
%\mathcal{U}_{-}(\lambda +\eta /2)\tilde{G}_{m}(-\lambda )  \notag \\
%&=&\sinh (2\lambda -2\eta )\tilde{G}_{m}^{-1}(-\lambda )\mathcal{U}_{-}(\eta
%/2-\lambda )\mathcal{U}_{-}(\lambda +\eta /2)\tilde{G}_{m}(-\lambda )  \notag
%\\
%&=&\tilde{G}_{m}^{-1}(-\lambda )\det_{q}\mathcal{U}_{-}(\lambda )\tilde{G}%
%_{m}(-\lambda )  \notag \\
%&=&\det_{q}\mathcal{U}_{-}(\lambda ).
%\end{eqnarray}%
Now the representations (\ref{gauge-q-det-A})  and (\ref{gauge-q-det-D})  for the quantum determinant  follow directly from these expressions. 
\end{proof}

%\subsubsection{Symmetry relations between gauge transformed boundary
%operators}

It is also easy to establish a $\beta$-parity relation for the gauged monodromy matrix
\begin{proposition}
The following identity holds:%
\begin{equation}
\mathcal{U}_{-}(\lambda |-\beta +2)=\sigma ^{x}\mathcal{U}_{-}(\lambda
|\beta)\sigma ^{x}  \label{U-gauge-symm}
\end{equation}%
or for the matrix elements: 
\begin{equation}
\mathcal{B}_{-}(\lambda |\beta)=\mathcal{C}_{-}(\lambda |-\beta +2),\text{ \
\ }\mathcal{A}_{-}(\lambda |\beta)=\mathcal{D}_{-}(\lambda |-\beta +2).
\label{B-to-C-identity}
\end{equation}
\end{proposition}

\begin{proof}
The proof is a trivial consequence of the following simple identities:%
\begin{equation}
Y(\lambda|\beta )=X(\lambda|-\beta ).
\end{equation}%
%e.g. we have that:%
%\begin{eqnarray}
%\widehat{\mathcal{B}}_{-}(\lambda |\beta) &=&\tilde{Y}_{m-1}(\lambda -\eta /2)%
%\mathcal{U}_{-}(\lambda )Y_{m-1}(\eta /2-\lambda )  \notag \\
%&=&\tilde{X}_{(-m-2\beta +2)-1}(\lambda -\eta /2)\mathcal{U}_{-}(\lambda
%)X_{(-m-2\beta +2)-1}(\eta /2-\lambda )  \notag \\
%&=&\widehat{\mathcal{C}}_{-}(\lambda |-m-2\beta +2).
%\end{eqnarray}
\end{proof}

\subsection{Boundary transfer matrix and gauged
operators}

It is possible to write the boundary transfer matrix $\mathcal{T}(\lambda )$ for the
most general boundary conditions in terms of the gauged boundary operators. First we need to transform in an appropriate way the boundary matrix $K_+(\lambda)$. There are two possible ways to do it (we will call them left and right $K_+$ matrices). 

We introduce two new  vectors
\begin{equation}
\hat{X}(\lambda|\beta+2)=e^\eta\frac{\sinh(\beta-1)\eta}{\sinh\beta\eta} X(\lambda|\beta+2),\quad
\hat{Y}(\lambda|\beta-2)=e^\eta\frac{\sinh(\beta+1)\eta}{\sinh\beta\eta} Y(\lambda|\beta-2).
\end{equation}

 Then we can define the following $2\times2$ matrix
\begin{align}
&K_{+}^{(L)}(\lambda |\beta)=\nonumber\\
&\left( 
\begin{array}{ll}
 \tilde{Y}(\eta /2-\lambda|\beta-2
)K_{+}(\lambda )\hat{X}(\lambda -\eta /2|\beta+2) & 
 \tilde{Y}(\eta /2-\lambda|\beta )K_{+}(\lambda )\hat{Y}%
(\lambda -\eta /2|\beta-2) \\ 
 \tilde{X}(\eta /2-\lambda|\beta
)K_{+}(\lambda )\hat{X}(\lambda -\eta /2|\beta+2) &  \tilde{X}(\eta /2-\lambda|\beta+2 )K_{+}(\lambda )\hat{Y}%
(\lambda -\eta /2|\beta-2)%
\end{array}%
\right),
\end{align}
the right $K_+$ matrix is defined in a similar way
\begin{equation}
K_{+}^{(R)}(\lambda |\beta)=\left( 
\begin{array}{ll}
 \bar{Y}(\eta /2-\lambda|\beta
)K_{+}(\lambda )X(\lambda -\eta /2|\beta) & 
 \bar{Y}(\eta /2-\lambda|\beta )K_{+}(\lambda
)Y(\lambda -\eta /2|\beta-2) \\ 
\bar{X}(\eta /2-\lambda|\beta
)K_{+}(\lambda )X(\lambda -\eta /2|\beta+2) &  \bar{X}(\eta /2-\lambda|\beta )K_{+}(\lambda
)Y(\lambda -\eta /2|\beta)%
\end{array}%
\right) ,
\end{equation}%
 The explicit expressions for these two matrices are given in the Appendix.
 
\begin{lemma}
The boundary transfer matrix admits the two following representations in terms of the gauged generators:%
\begin{align}
e^{-\lambda+\eta/2}\mathcal{T}(\lambda )& =K_{+}^{(L)}(\lambda |\beta-1)_{11}\mathcal{A}_{-}(\lambda
|\beta)+K_{+}^{(L)}(\lambda |\beta-1)_{22}\mathcal{D}_{-}(\lambda |\beta)  \notag \\
& +K_{+}^{(L)}(\lambda |\beta-1)_{21}\mathcal{B}_{-}(\lambda
|\beta-2)+K_{+}^{(L)}(\lambda |\beta-1)_{12}\mathcal{C}_{-}(\lambda |\beta+2),
\label{T-decomp-L}
\end{align}%
and%
\begin{align}
e^{-\lambda+\eta/2}\mathcal{T}(\lambda )& =K_{+}^{(R)}(\lambda |\beta-1)_{11}\mathcal{A}_{-}(\lambda
|\beta)+K_{+}^{(R)}(\lambda |\beta-1)_{22}\mathcal{D}_{-}(\lambda |\beta)  \notag \\
& +K_{+}^{(R)}(\lambda |\beta-1)_{21}\mathcal{B}_{-}(\lambda
|\beta+2)+K_{+}^{(R)}(\lambda |\beta-1)_{12}\mathcal{C}_{-}(\lambda |\beta).
\label{T-decomp-R}
\end{align}%
\end{lemma}

\begin{proof}[Proof]
To prove the expression (\ref{T-decomp-L})
we  introduce a new gauge matrix
\begin{equation}
\widehat{G}(\lambda|\beta)=\left(
\hat{X}(\lambda|\beta+2),\hat{Y}(\lambda|\beta-2)
\right).
\end{equation}
  It is not difficult to check that 
\begin{equation}
\widehat{G}^{-1}(\lambda|\beta)=\left(\begin{array}{l}
\tilde{Y}(\lambda|\beta-2)\\ \tilde{X}(\lambda|\beta+2)\end{array}
\right).
\end{equation}
Now we can rewrite  the right hand side of (\ref{T-decomp-L}) as follows
\begin{align}
e^{\lambda-\eta/2}\Big( \mathcal{A}_{-}(\lambda |\beta)K_{+}^{(L)}(\lambda |\beta-1)_{11}+&\mathcal{B}%
_{-}(\lambda |\beta-2)K_{+}^{(L)}(\lambda |\beta-1)_{21}
\notag \\
+&\mathcal{D}_{-}(\lambda
|\beta)K_{+}^{(L)}(\lambda |\beta-1)_{22}+\mathcal{C}_{-}(\lambda
|\beta+2)K_{+}^{(L)}(\lambda |\beta-1)_{12}\Big)  \notag \\
  = \tilde{Y}(\lambda -\eta /2|\beta-3)%
\mathcal{U}_{-}(\lambda )\,&K_{+}(\lambda )\hat{X}(\lambda -\eta
/2|\beta+1)  \notag \\+&\tilde{X}(\lambda -\eta /2|\beta+1)\mathcal{U}%
_{-}(\lambda )K_{+}(\lambda )\hat{Y}(\lambda -\eta /2|\beta-3)  \notag \\
  = \mathrm{tr}_{0}\{%
\widehat{G}^{-1}(\lambda-\eta/2|\beta-1) &\mathcal{U}_{-}(\lambda )K_{+}(\lambda )\widehat{G}(\lambda-\eta/2|\beta-1) \}  \notag \\
& =\text{tr}_{0}\{
 \mathcal{U}_{-}(\lambda )K_{+}(\lambda )\}  = \mathcal{T}(\lambda ).
\end{align}%
The expression (\ref{T-decomp-R}) can be proved in a similar way.
\end{proof}

\begin{proposition}
The most general transfer matrix can be written in the following form 
\begin{align}
\mathcal{T}(\lambda )& =\mathsf{a}_{+}(\lambda|\beta-1 )\mathcal{A}_{-}(\lambda
|\beta)+\mathsf{a}_{+}(-\lambda|\be-1 )\mathcal{A}_{-}(-\lambda
|\beta)\notag\\
&+K_{+}^{(L)}(\lambda |\beta-1)_{21}\mathcal{B}_{-}(\lambda
|\beta-2)+K_{+}^{(L)}(\lambda |\beta-1)_{12}\mathcal{C}_{-}(\lambda |\beta+2), \\
\mathcal{T}(\lambda )& =\mathsf{d}_{+}(\lambda|\beta-1 )\mathcal{D}_{-}(\lambda
|\beta)+\mathsf{d}_{+}(-\lambda|\be-1 )\mathcal{D}_{-}(-\lambda
|\beta)\notag\\
&+K_{+}^{(R)}(\lambda |\beta-1)_{21}\mathcal{B}_{-}(\lambda
|\beta)+K_{+}^{(R)}(\lambda |\beta-1)_{12}\mathcal{C}_{-}(\lambda |\beta),
\end{align}%
where we have defined:%
\begin{align}
\mathsf{a}_{+}(\lambda|\beta )& =\frac{\sinh (2\lambda +\eta )} {\sinh 2\lambda \sinh( \beta-1) \eta \sinh
\zeta _{+}}\Big[ \sinh \zeta _{+}\cosh
(\lambda -\eta /2)\sinh (\lambda +\eta /2+\beta \eta)  \notag \\
&  -\left( \cosh \zeta _{+}\sinh (\lambda -\eta /2)\cosh (\lambda
+\eta /2+\beta \eta) +\kappa _{+}\sinh (2\lambda -\eta )\sinh (\tau
_{+}+\alpha\eta +2\eta) \right) \Big] \label{a+}\\
\mathsf{d}_{+}(\lambda|\beta )& =\frac{\sinh (2\lambda +\eta )}{ \sinh 2\lambda \sinh (\beta-1) \eta \sinh
\zeta _{+}}\Big[ \sinh \zeta _{+}\cosh
(\lambda -\eta /2)\sinh (-\lambda -\eta /2+\beta \eta)   \notag \\
&  -\left( \cosh \zeta _{+}\sinh (\lambda -\eta /2)\cosh (-\lambda
-\eta /2+\beta \eta) +\kappa _{+}\sinh (2\lambda -\eta )\sinh (\tau
_{+}+\alpha \eta) \right) \Big]. \label{d+}
\end{align}
\end{proposition}
To prove this proposition one should use the properties of the gauged operators and the explicit form of the $K_+$ matrices given in Appendix A.

\subsection{Reference states}

The ferromagnetic left and right states 
$$\bra{0}=\otimes _{n=1}^{\mathsf{N}}\Big( 1,0\Big)_n,\quad \ket{0}=\otimes _{n=1}^{\mathsf{N}}\left(\begin{array}{cc}1\\0\end{array}\right)_n$$
 are no more the highest weight vectors for the spin chains with non-diagonal boundaries and it is the reason why the Bethe ansatz does not work directly for this case. However using the gauge transformation we can define new reference states which can be used in the SOV framework.
 
We  define the following left reference state:%
\begin{equation}
\langle \beta|\equiv \otimes _{n=1}^{\mathsf{N}}\left( -1,e^{-\alpha \eta +(%
\mathsf{N}-n+\beta )\eta -\xi _{n}}\right) _{(n)}=N_{\beta}\bra{0}\prod_{n=1}^{%
\mathsf{N}}\bar{G}^{-1}_n(\xi _{n}|\be+\mathsf{N}-n),%
  \label{Left-B-ref}
\end{equation}%
where $\bar{G}_{n}^{-1 }(\xi _{n})$ is the gauge transformation acting in the {\it local quantum space}
 $\mathcal{H}_n$  and $N_{\beta}$ is a
normalization factor
\begin{equation}
N_{\beta}=2^{\mathsf{N}}e^{-\alpha \mathsf{N}\eta }\prod_{n=1}^{\mathsf{N}}\sinh
(\mathsf{N}-n+\beta )\eta .
\end{equation}
\begin{proposition}
The state $\langle \beta|$ is a simultaneous $B(\lambda |\beta)$ and $\bar{B}%
(\lambda |\beta)$ left reference state:%
\begin{eqnarray}
\langle \beta|B(\lambda |\beta) &=&\langle \beta|\bar{B}(\lambda |\beta)=0,  \label{Id-left-ref1} \\
\langle \beta|A(\lambda |\beta) &=&\frac{\sinh (\mathsf{N}+\beta )\eta }{\sinh
\beta \eta }\prod_{n=1}^{\mathsf{N}}\sinh (\lambda -\xi _{n}+\eta
/2)\langle \beta-1| ,\\
\langle \beta|D(\lambda |\beta) &=&\prod_{n=1}^{\mathsf{N}}\sinh (\lambda -\xi
_{n}-\eta /2)\langle \beta+1|, \\
\langle \beta|\bar{A}(\lambda |\beta) &=&\frac{\sinh \beta \eta }{\sinh (%
\mathsf{N}+\beta )\eta }\prod_{n=1}^{\mathsf{N}}\sinh (\lambda +\xi
_{n}+\eta /2)\langle \beta+1| ,\\
\langle \beta|\bar{D}(\lambda |\beta) &=&\prod_{n=1}^{\mathsf{N}}\sinh (\lambda +\xi
_{n}-\eta /2)\langle \beta-1| . \label{Id-left-ref5}
\end{eqnarray}
\end{proposition}

%\begin{proof}
The proposition can be checked for local $R$-matrices by direct computation.
%\end{proof}

%\subsubsection{\label{Right-ref-section}Simultaneous $C(\protect\lambda |\beta)$
%and $\bar{C}(\protect\lambda |\beta)$ bulk right reference state}

Similarly we can define the right reference state
\begin{equation}
|\beta\rangle \equiv \otimes _{n=1}^{\mathsf{N}}%
\begin{pmatrix}
e^{-\alpha \eta -(\mathsf{N}-n+\beta )\eta -\xi _{n}} \\ 
1%
\end{pmatrix}%
=\prod_{n=1}^{\mathsf{N}}\bar{G}_n(\xi
_{n}|\beta+\mathsf{N}-n)\ket{0},  \label{Right-C-ref}
\end{equation}%
and the
following proposition holds:

\begin{proposition}
\label{Right-ref}The state $|\beta+1\rangle $ is a simultaneous $C(\lambda |\beta)$
and $\bar{C}(\lambda |\beta)$ right reference state:
\begin{eqnarray}
C(\lambda |\beta)|\beta+1\rangle &=&\bar{C}(\lambda |\beta)|\beta+1\rangle=0 ,%
\text{ } \\
A(\lambda |\beta)|\beta+1\rangle &=&\prod_{n=1}^{\mathsf{N}}\sinh (\lambda -\xi
_{n}+\eta /2)|\beta+2\rangle , \\
D(\lambda |\beta)|\beta+1\rangle &=&\frac{\sinh \eta (\mathsf{N}+\beta )}{\sinh
\eta \beta }\prod_{n=1}^{\mathsf{N}}\sinh (\lambda -\xi _{n}-\eta
/2)|\beta\rangle , \\
\bar{A}(\lambda |\beta)|\beta+1\rangle &=&\prod_{n=1}^{\mathsf{N}}\sinh (\lambda
+\xi _{n}+\eta /2)|\beta\rangle , \\
\bar{D}(\lambda |\beta)|\beta+1\rangle &=&\frac{\sinh \eta \beta }{\sinh \eta (%
\mathsf{N}+\beta )}\prod_{n=1}^{\mathsf{N}}\sinh (\lambda +\xi _{n}-\eta
/2)|\beta+2\rangle .
\end{eqnarray}
\end{proposition}

\section{SOV representations of the gauge transformed reflection algebra}

In this section we construct explicitly the SOV representation of the gauged reflection algebra.  In general it is associated to the construction of the eigenstates of the operators $\mathcal{B}$ (or $\mathcal{C}$). However the gauge transformation and the particular structure of the reference states leads to a slightly different result. Instead of the eigenstates we construct right and left {\it pseudo-eigenstates} for these operators. More precisely, for any generic value of $\beta$ we will construct a basis in the Hilbert space $\mathcal{H}$
\[\bra{\beta,\mathbf{h}},\quad  \mathbf{h}\equiv (h_{1},...,h_{\mathsf{N}}),\quad h_{j}\in \{0,1\},\]
formed by states that we will call left pseudo-eigenstates of $\mathcal{B}_{-}(\lambda |\beta)$ if they satisfy the identities
\begin{equation}
\langle \beta,\mathbf{h}|\mathcal{B}_{-}(\lambda |\beta)=\mathsf{B}_{\mathbf{h}}(\lambda |\beta)\langle \beta-2,\mathbf{h}|,
%\label{right-B-eigen-cond}
\end{equation}
where for all the possible $\mathbf{h}$ the $\mathsf{B}_{-}(\lambda |\beta)$ are the pseudo-eigenvalues of $\mathcal{B}_{-}(\lambda |\beta)$, central elements in the algebra.
Similarly we can define the basis of right pseudo-eigenstates. 

The  results of this section can be summarized in  the following theorem. 

\begin{theorem}
\label{Th1}Let the inhomogeneities $\{\xi _{1},...,\xi _{\mathsf{N}}\}\in 
\mathbb{C}$ $^{\mathsf{N}}$ satisfy the following conditions:%
\begin{equation}
\xi _{a}\neq \xi _{b}+r\eta \text{ \ }\forall a\neq b\in \{1,...,\mathsf{N}%
\}\,\,\text{and\thinspace \thinspace }r\in \{-1,0,1\},  \label{E-SOV}
\end{equation}%
then:

I$_{b}$) for any  $\alpha ,\beta \in \mathbb{C}$ such that for any integer $k$%
\begin{equation}
(\alpha -\beta )\eta \neq (\mathsf{N}-1)\eta -\tau _{-}-(-1)^k (\alpha
_{-}+\beta _{-})+i\pi k,  \label{NON-nilp-B-L}
\end{equation}%
the one parameter family of the gauge transformed generators of the
reflection algebra $\mathcal{B}_{-}(\lambda |\beta)$\ is left
pseudo-diagonalizable and its pseudo-spectrum is simple.

II$_{b}$) for any fixed  $\alpha ,\beta \in \mathbb{C}$ such that  for any integer $k$
\begin{equation}
(\alpha -\beta )\eta \neq -(\mathsf{N}+1)\eta -\tau _{-}-(-1)^k (\alpha
_{-}+\beta _{-})+i\pi k,  \label{NON-nilp-B-R}
\end{equation}%
the one parameter family of the gauge transformed generators of the
reflection algebra $\mathcal{B}_{-}(\lambda |\beta)$\ is right
pseudo-diagonalizable and its pseudo-spectrum is simple.

I$_{c}$) for any fixed $m\in \mathbb{Z}$, $\alpha ,\beta \in \mathbb{C}$:%
\begin{equation}
(\alpha +\beta )\eta \neq (\mathsf{N}+1) \eta -\tau _{-}-(-1)^k
(\alpha _{-}+\beta _{-})+i\pi k,  \label{NON-nilp-C-L}
\end{equation}%
the one parameter family of the gauge transformed generators of the
reflection algebra $\mathcal{C}_{-}(\lambda |\beta)$\ is left
pseudo-diagonalizable and its pseudo-spectrum is simple.

II$_{c}$) for any  $\alpha ,\beta \in \mathbb{C}$ such that  for any integer $k$:%
\begin{equation}
(\alpha +\beta )\eta \neq -(\mathsf{N}-1)\eta -\tau _{-}-(-1)^k (\alpha
_{-}+\beta _{-})+i\pi k,  \label{NON-nilp-C-R}
\end{equation}%
the one parameter family of the gauge transformed generators of the
reflection algebra $\mathcal{C}_{-}(\lambda |\beta)$\ is right
pseudo-diagonalizable and its pseudo-spectrum is simple.

In all these cases we can construct a SOV representation of the gauge
transformed reflection algebra.
\end{theorem}

The proof and some necessary clarifications of the statements contained in
this theorem are given by the explicit constructions of the SOV representation in
the next subsections. In fact, we do these constructions explicitly only for
the cases I$_{b}$) and II$_{b}$) as for the cases I$_{c}$) and II$_{c}$)
these constructions can be induced from the others due to the symmetries.

\subsection{$\mathcal{B}_{-}(\la|\beta)$-SOV representations of the gauge
transformed reflection algebra}

\subsubsection{Left $\mathcal{B}_{-}(\la|\beta)$-SOV representations of the gauge
transformed reflection algebra}

In this subsection we construct %the left SOV-representations by constructing
the left $\mathcal{B}_{-}(\lambda |\beta)$-pseudo-eigenbasis.

\begin{theorem}
\underline{Left $\mathcal{B}_{-}(\lambda |\beta)$ SOV-basis} \ The
following states:%
\begin{equation}
\langle \beta,h_{1},...,h_{\mathsf{N}}|= %\frac{1}{\text{\textsc{n}}_{\beta+2}}%
\langle \beta|\prod_{n=1}^{\mathsf{N}}\left( \frac{\mathcal{A}_{-}(\eta /2-\xi
_{n}|\beta+2)}{\mathsf{A}_{-}(\eta /2-\xi _{n})}\right) ^{h_{n}},
\label{D-left-eigenstates}
\end{equation}%
where $\langle \beta|$ is the state defined in (\ref{Left-B-ref}) and the function $\mathsf{A}_-$ is given by (\ref{eigenA}).
%and \textsc{n}$_{m+2}$ a normalization to be fixed.
 If ( \ref{E-SOV}%
)  and ( \ref{NON-nilp-B-L})  are satisfied, these
states define a basis of $\mathcal{H}$ formed out of
pseudo-eigenstates of $\mathcal{B}_{-}(\lambda |\beta)$:%
\begin{equation}
\langle \beta,\mathbf{h}|\mathcal{B}_{-}(\lambda |\beta)=\mathsf{B}_{\mathbf{h}}(\lambda |\beta)\langle \beta-2,\mathbf{h}|,
\label{right-B-eigen-cond}
\end{equation}%
where $\langle \beta,\mathbf{h}|= \langle \beta,h_{1},...,h_{\mathsf{N}}|$,\,
  \,$\mathbf{h}= (h_{1},...,h_{\mathsf{N}})$,\, $h_{j}\in \{0,1\}$ and%
\begin{align}
\mathsf{B}_{\mathbf{h}}(\lambda |\beta)=& \left( -1\right) ^{%
\mathsf{N}}e^{(\beta +\mathsf{N})\eta }a_{\mathbf{h}%
}(\lambda )a_{\mathbf{h}}(-\lambda )\notag\\
&\times\frac{\sinh (2\lambda -\eta )\left(
2\kappa _{-}\sinh \left[ (\mathsf{N}+\beta -\alpha -1)\eta -\tau _{-}%
\right] -e^{\zeta _{-}}\right) }{2%(\text{\textsc{n}}_{m+2}/\text{\textsc{n}}_{m})
\sinh \zeta _{-}\sinh (\mathsf{N}+\beta )\eta },
\end{align}%
with%
\begin{equation}\label{a_h}
a_{\mathbf{h}}(\lambda )= \prod_{n=1}^{\mathsf{N}}\sinh (\lambda
-\xi _{n}-(h_{n}-\frac{1}{2})\eta ).
\end{equation}%
\end{theorem}
\begin{proof}[Proof]
It is worth writing explicitly the (boundary-bulk) decomposition of the
gauge transformed reflection algebra generator (\ref{boundary-bulkBD})%
\begin{align}
e^{\lambda -\eta /2}\mathcal{B}_{-}(\lambda |\beta)& =K_{-}(\lambda
|\beta)_{12}A(\lambda |\beta)\bar{D}(\lambda |\beta-1)+K_{-}(\lambda |\beta)_{11}A(\lambda
|\beta)\bar{B}(\lambda |\beta-1)  \notag \\
& +K_{-}(\lambda |\beta)_{21}B(\lambda |\beta)\bar{B}(\lambda |\beta-1)+K_{-}(\lambda
|\beta)_{22}B(\lambda |\beta)\bar{D}(\lambda |\beta-1).
\end{align}%
Then, the formulae (\ref{Id-left-ref1}- \ref{Id-left-ref5})
 imply that $\langle \beta|$ is a $\mathcal{B}%
_{-}(\lambda |\beta)$-pseudo-eigenstate with non-zero eigenvalue:%
\begin{equation}
\langle \beta|\mathcal{B}_{-}(\lambda )= \mathsf{B}_{\mathbf{%
0}}(\lambda|\beta )\langle \beta-2|,
\end{equation}%
where:%
\begin{equation}
\mathsf{B}_{\mathbf{0}}(\lambda |\beta)=%\frac{\text{\textsc{n}}_{m}%
%}{\text{\textsc{n}}_{m+2}}
\left( -1\right) ^{\mathsf{N}}e^{-\lambda +\eta
/2}K_{-}(\lambda |\beta)_{12}a_{\mathbf{0}}(\lambda )a_{\mathbf{0}%
}(-\lambda ),
\end{equation}%
$a_{\mathbf{0}}(\lambda )$ is given by (\ref{a_h}) for all $h_j=0$ and:%
\begin{equation}
e^{-\lambda +\eta /2}K_{-}(\lambda |\beta)_{12}=\frac{e^{(\beta +\mathsf{N})\eta }\sinh
(2\lambda -\eta )(2\kappa _{-}\sinh \left[ (\mathsf{N}+\beta -\alpha -1)\eta -\tau
_{-}\right] -e^{\zeta _{-}})}{2\sinh (\mathsf{N}+\beta )\eta \sinh \zeta _{-}}.
\end{equation}%
Now by using the reflection algebra commutation relations we can follow step
by step the proof given in \cite{Nic12b} to prove the validity of (\ref%
{right-B-eigen-cond}). Under the condition (\ref{E-SOV}),
these relations also imply that the set of states $\langle \beta$, $\mathbf{h}|$
forms a set of 2$^{\mathsf{N}}$ independent states, i.e. a $\mathcal{B}%
_{-}(\lambda |\beta)$-pseudo-eigenbasis of $\mathcal{H}$. The
action of $\mathcal{A}_{-}(\zeta _{b}^{(h_{b})}|\beta+2)$ for $b\in \{1,...,2%
\mathsf{N}\}$ follows by the definition of the states $\langle \beta,\mathbf{h}%
|$, the reflection algebra commutation relations (\ref{CMR-AB-Left}), the quantum determinant relations and the conditions:%
\begin{equation}
\langle \beta|\mathcal{A}_{-}(\xi _{n}-\eta /2|\beta+2)=0, \ \ \langle \beta|%
\mathcal{A}_{-}(\eta /2-\xi _{n}|\beta+2)\neq 0
\end{equation}%
which trivially follows from the boundary-bulk decomposition (\ref{boundary-bulkAC})%
\begin{align}
e^{\lambda -\eta /2}\mathcal{A}_{-}(\lambda |\beta+2)& =\bar{K}_{-}(\lambda
|\beta)_{11}A(\lambda |\beta)\bar{A}(\lambda |\beta+1)+\bar{K}_{-}(\lambda
|\beta)_{12}A(\lambda |\beta)\bar{C}(\lambda |\beta+1)  \notag \\
& +\bar{K}_{-}(\lambda |\beta)_{21}B(\lambda |\beta)\bar{A}(\lambda |\beta+1)+\bar{K}%
_{-}(\lambda |\beta)_{22}B(\lambda |\beta)\bar{C}(\lambda |\beta+1).
\end{align}%
It is important to point out that the states $\langle \beta$, 
$\mathbf{h}|$ are well defined non-zero states and their definition
does not depend on the order of operator $\mathcal{A}_{-}(\eta/2-\xi
_{b}|\beta+2)$ in their definition as it  follows from the
commutation relations (\ref{CMR-AA-BC}).
\end{proof}

\begin{theorem}
The action of the 
reflection algebra generators $\mathcal{A}_{-}(\lambda |\beta+2)$ on the generic state $\langle \beta$, $\mathbf{h}|$,  is given by the following expression 
\begin{align}
\langle \beta,\mathbf{h}|\mathcal{A}_{-}(\lambda |\beta+2)& =\sum_{a=1}^{2%
\mathsf{N}}\frac{\sinh (2\lambda -\eta )\sinh (\lambda +\zeta _{a}^{(h_{a})})%
}{\sinh (2\zeta _{a}^{(h_{a})}-\eta )\sinh 2\zeta _{a}^{(h_{a})}}\notag\\
&\times\prod 
_{\substack{ b=1  \\ b\neq a\text{ mod}\mathsf{N}}}^{\mathsf{N}}\frac{\cosh
2\lambda -\cosh 2\zeta _{b}^{(h_{b})}}{\cosh 2\zeta _{a}^{(h_{a})}-\cosh
2\zeta _{b}^{(h_{b})}}\mathsf{A}_{-}(\zeta _{a}^{(h_{a})})  
 \langle \beta,\mathbf{h}|T_{a}^{-\varphi
_{a}}
 \notag \\
& +\det_{q}M(0)\cosh (\lambda -\eta /2)\prod_{b=1}^{\mathsf{N}}\frac{%
\cosh 2\lambda -\cosh 2\zeta _{b}^{(h_{b})}}{\cosh \eta -\cosh 2\zeta
_{b}^{(h_{b})}}\langle \beta,\mathbf{h}|  \notag \\
& +(-1)^{\mathsf{N}+1}\coth \zeta _{-}\det_{q}M(i\pi /2)\sinh (\lambda -\eta
/2)\prod_{b=1}^{\mathsf{N}}\frac{\cosh 2\lambda -\cosh 2\zeta _{b}^{(h_{b})}%
}{\cosh \eta +\cosh 2\zeta _{b}^{(h_{b})}}\langle \beta,\mathbf{h}|,
\label{L-SOV A-}
\end{align}%
where $h_{n+\mathsf{N}}\equiv h_n\in \{0,1\}$, and%
\begin{equation}
\zeta _{n}^{(h_{n})} =\varphi _{n}\left[ \xi _{n}+(h_{n}-\frac{1}{2})\eta %
\right]\quad \forall n\in
\{1,...,2\mathsf{N}\}, 
\end{equation}
\begin{equation}
\varphi _{a} =1\quad \text{ for }\quad a\le\mathsf{N}\quad \text{  and \ }\quad\varphi _{a} =-1 \quad\text{ for }\quad a>\mathsf{N},
\end{equation}%
and:%
\begin{equation}
\langle \beta,h_{1},...,h_{a},...,h_{\mathsf{N}}|T_{a}^{\pm }=\langle
\beta,h_{1},...,h_{a}\pm 1,...,h_{\mathsf{N}}|.
\end{equation}%
\smallskip
\end{theorem}

\begin{proof}[Proof]
Using the identities:%
\begin{equation}
\mathcal{U}_{-}(\eta /2)=\det_{q}M(0)\text{ }I_{0},\text{ \ \ }\mathcal{U}%
_{-}(\eta /2+i\pi /2)=i\coth \zeta _{-}\det_{q}M(i\pi /2)\text{ }\sigma
_{0}^{z},  \label{U-identities}
\end{equation}%
and%
\begin{equation}
\tilde{Y}(0\vert \beta-1)X(0\vert \beta+1)=1,\text{ \ }\tilde{Y}(i\pi /2\vert \beta-1)\sigma
_{0}^{z}X(-i\pi /2\vert \beta+1)=-1
\end{equation}%
and taking into account  that $\mathcal{A}_{-}(\lambda |\beta)$ has the following
functional dependence with respect to $\lambda $:%
\begin{equation}
\mathcal{A}_{-}(\lambda |\beta)=\sum_{a=0}^{2\mathsf{N}+1}e^{\left( 2a-2\mathsf{N%
}+1\right) \lambda }\mathcal{A}_{m,a}
\end{equation}%
we get the following interpolation formula for the action on $\langle \beta$, 
$\mathbf{h}|$:%
\begin{align*}
\langle \beta,\mathbf{h}|\mathcal{A}_{-}(\lambda |\beta+2)& =\sum_{a=1}^{2%
\mathsf{N}}\frac{\sinh(2\la-\eta)}{\sinh(2\zeta_a^{(h_a)}-\eta)}\prod_{\substack{ b=1  \\ b\neq a}}^{2\mathsf{N}}\frac{\sinh
(\lambda -\zeta _{b}^{(h_{b})})}{\sinh (\zeta _{a}^{(h_{a})}-\zeta
_{b}^{(h_{b})})}\mathsf{A}_{-}(\zeta _{a}^{(h_{a})})\langle \beta,\mathbf{h%
}|T_{a}^{-\varphi _{a}}  \notag \\
& +\det_{q}M(0)\cosh(\la-\eta/2)\prod_{ b=1}^{2\mathsf{N}%
}\frac{\sinh (\lambda -\zeta _{b}^{(h_{b})})}{\sinh (\eta/2-\zeta _{b}^{(h_{b})})}\langle \beta,\mathbf{h}|  \notag \\
& \coth \zeta _{-}\det_{q}M(i\pi /2)\sinh(\la-\eta/2)\prod_{b=1}^{2\mathsf{N}}\frac{\sinh
(\lambda -\zeta _{b}^{(h_{b})})}{\sinh(\eta/2+i\pi/2-\zeta
_{b}^{(h_{b})})}\langle \beta,\mathbf{h}|.
\end{align*}%
Then, it is a simple exercise to rewrite this in the form (\ref{L-SOV
A-}). 
\end{proof}

\subsubsection{Right $\mathcal{B}_{-}(\la|\beta)$-SOV representations of the gauge
transformed reflection algebra}

\begin{theorem}
\underline{Right $\mathcal{B}_{-}(\lambda |\beta)$ SOV-basis} \ We
define the states:%
\begin{equation}
|\beta,h_{1},...,h_{\mathsf{N}}\rangle = %\frac{1}{\text{\textsc{n}}_{m}}%
\prod_{n=1}^{\mathsf{N}}\left( \frac{\mathcal{D}_{-}(\xi _{n}+\eta /2|\beta)}{%
%k_{n}^{(m)}
f_n(\beta)\mathsf{A}_{-}(\eta /2-\xi _{n})}\right) ^{(1-h_{n})}%
|-\beta+2\rangle ,  \label{D-right-eigenstates}
\end{equation}%
where:%
\begin{equation}
f_n(\beta)= \frac{\sinh (2\xi _{n}+\eta )\sinh \beta \eta }{\sinh
(2\xi _{n}-\eta )\sinh (2\xi _{n}+\beta \eta )},
\end{equation}
and $h_{n}\in \{0,1\},$ $n\in \{1,...,\mathsf{N}\}$. If (\ref{E-SOV}%
) and (\ref{NON-nilp-B-R})
are satisfied, then 
 this set of states defines a basis of $\mathcal{H}$ and they are $\mathcal{B}_{-}(\lambda |\beta)$ right pseudo-eigenstates:
\begin{equation}
\mathcal{B}_{-}(\lambda |\beta)|\beta,\mathbf{h}\rangle =|\beta+2,\mathbf{h}%
\rangle \mathsf{\bar B}_{\mathbf{h}}(\lambda|\beta ),
\label{left-B-eigen-cond}
\end{equation}%
where:%
\begin{align}
\mathsf{\bar B}_{\mathbf{h}}(\lambda |\beta)=& \left( -1\right)
^{\mathsf{N}}e^{(\beta -\mathsf{N})\eta } \prod_{n=1}^{\mathsf{N}}
\left(\frac {f_n(\beta+2)}{f_n(\beta)}%
%}k_{n}^{(m)}/k_{n}^{(m+2)}
\right)^{1-h_n}a_{\mathbf{h}}(\lambda )a_{%
\mathbf{h}}(-\lambda )\notag\\
&\times
\frac{\sinh (2\lambda -\eta
)\left( 2\kappa _{-}\sinh \left[ (\beta -(1+\mathsf{N}+\alpha ))\eta -\tau
_{-}\right] -e^{\zeta _{-}}\right) }{2\sinh \zeta _{-}\sinh \beta \eta %(%
%\text{\textsc{n}}_{m}/\text{\textsc{n}}_{m+2})
}
.
\end{align}%
\end{theorem}
\begin{proof}[Proof]
The proof is similar  to the one for the left SOV basis.  First  we prove that $%
|-\beta+2\rangle $ is a right $\mathcal{B}_{-}(\lambda |\beta)$ pseudo-eigenstate. From the Proposition \ref{Right-ref} and the 
boundary-bulk decomposition (\ref{boundary-bulkAC}):%
\begin{align*}
e^{\lambda -\eta /2}\mathcal{C}_{-}(\lambda |\beta)& =\bar{K}_{-}(\lambda
|\beta-2)_{21}D(\lambda |\beta-2)\bar{A}(\lambda |\beta-1)+\bar{K}_{-}(\lambda
|\beta-2)_{22}D(\lambda |\beta-2)\bar{C}(\lambda |\beta-1)  \notag \\
& +\bar{K}_{-}(\lambda |\beta-2)_{12}C(\lambda |\beta-2)\bar{C}(\lambda |\beta-1)+\bar{%
K}_{-}(\lambda |\beta-2)_{11}C(\lambda |\beta-2)\bar{A}(\lambda |\beta-1).
\end{align*}%
It follows that the state $|\beta\rangle $ is a right $\mathcal{C}_{-}(\lambda
|\beta)$-pseudo-eigenstate; i.e. it holds:%
\begin{equation}
\mathcal{C}_{-}(\lambda |\beta)|\beta\rangle =|\beta-2\rangle \mathsf{C}%
_{\mathbf{0}}(\lambda|\beta )  \label{right-C-boundary-state}
\end{equation}%
where:%
\begin{equation}
\mathsf{C}_{\mathbf{0}}(\lambda|\beta )=\left( -1\right) ^{\mathsf{N}}e^{-\lambda
+\eta /2}\tilde{K}_{-}(\lambda |\beta)_{21}\frac{\sinh (\mathsf{N}+\beta-2
)\eta }{\sinh (\beta-2 )\eta }a_{\mathbf{1}}(\lambda )a_{%
\mathbf{1}}(-\lambda ),
\end{equation}%
and $a_{\mathbf{1}}(\lambda )$ is given by (\ref{a_h}) for all $h_j=1$ and%
\begin{equation*}
e^{-\lambda +\eta /2}\tilde{K}_{-}(\lambda |\beta)_{21}=\frac{e^{-(\beta
+\mathsf{N}-2)\eta }\sinh (2\lambda -\eta )\left( 2\kappa _{-}\sinh \left[ (\mathsf{N}+\beta
+\alpha -1)\eta +\tau _{-}\right] +e^{\zeta _{-}}\right) }{2\sinh \zeta
_{-}\sinh (\mathsf{N}+\beta -2)\eta }.
\end{equation*}%
Then from the identity (\ref{B-to-C-identity}), it follows
that the formula (\ref{right-C-boundary-state})  is equivalent
to the following one:%
\begin{equation}
\mathcal{B}_{-}(\lambda |\beta)|-\beta+2\rangle =|-\beta\rangle\, %
\mathsf{C}_{\mathbf{0}}(\lambda|-\beta+2 ).  \label{m-B-pseudoeigen}
\end{equation}%
Then by using the identities (\ref{m-B-pseudoeigen}) and the
commutation relations (\ref{BD-DB-CMR}) and the formulae: 
\begin{equation}
\mathcal{D}_{-}(-\xi _{n}-\eta /2|\beta)|-\beta+2\rangle =0,\quad%
\mathcal{D}_{-}(\xi _{n}+\eta /2|\beta)|-\beta+2\rangle \neq 0,
\label{D-on-m-B-pseudoeigen}
\end{equation}%
the states (\ref{D-right-eigenstates})  are proved to be
non-zero $\mathcal{B}_{-}(\lambda |\beta)$-pseudo-eigenstates with
pseudo-eigenvalues $\mathsf{\bar B}_{\mathbf{h}}(\lambda|\beta )$
which  form a basis of $\mathcal{H}$.
\end{proof}
To define the action of the operators $\mathcal{D}_{-}(\lambda |\beta)$ on the generic state $|\beta, \mathbf{h}\rangle $ we will need to introduce a set of values
\begin{equation}
\mathsf{D}_{-}(\zeta _{a}^{(h_{a})})=\Big[f_a(\beta))\Big]^{\varphi_a}\mathsf{A}_{-}(-\zeta _{a}^{(1-h_{a})}), \qquad a=1,\dots, 2 \mathsf{N}.
\end{equation}
It is important to underline that this set of values cannot be seen as values of some analytic  function $\mathsf{D}_-$, however to construct the SOV representation we will need only these points. 
\begin{theorem}
The action of the 
reflection algebra generators $\mathcal{D}_{-}(\lambda |\beta)$ on the generic state $|\beta, \mathbf{h}\rangle $,  can be written as follows%
\begin{align}
\mathcal{D}_{-}(\lambda |\beta)|\beta,\mathbf{h}\rangle  =&\sum_{a=1}^{2%
\mathsf{N}}T_{a}^{-\varphi _{a}}|\beta,\mathbf{h}\rangle \frac{%
\sinh (2\lambda -\eta )\sinh (\lambda +\zeta _{a}^{(h_{a})})}{\sinh (2\zeta
_{a}^{(h_{a})}-\eta )\sinh 2\zeta _{a}^{(h_{a})}}\notag \\
& \times\prod_{\substack{ b=1  \\ %
b\neq a\text{ mod}\mathsf{N}}}^{\mathsf{N}}\frac{\cosh 2\lambda -\cosh
2\zeta _{b}^{(h_{b})}}{\cosh 2\zeta _{a}^{(h_{a})}-\cosh 2\zeta
_{b}^{(h_{b})}}\mathsf{D}_{-}(\zeta _{a}^{(h_{a})})  \notag \\
& +|\beta,\mathbf{h}\rangle \det_{q}M(0)\cosh (\lambda -\eta
/2)\prod_{b=1}^{\mathsf{N}}\frac{\cosh 2\lambda -\cosh 2\zeta _{b}^{(h_{b})}%
}{\cosh \eta -\cosh 2\zeta _{b}^{(h_{b})}}  \notag \\
& +(-1)^{\mathsf{N}}|\beta,\mathbf{h}\rangle \coth \zeta
_{-}\det_{q}M(i\pi /2)\sinh (\lambda -\eta /2)\prod_{b=1}^{\mathsf{N}}\frac{%
\cosh 2\lambda -\cosh 2\zeta _{b}^{(h_{b})}}{\cosh \eta +\cosh 2\zeta
_{b}^{(h_{b})}},  \label{R-SOV D-}
\end{align}%
where:%
\begin{equation}
%\mathsf{D}_{-}(\zeta _{a}^{(h_{a})})=(k_{a}^{(m)})^{\varphi _{a}}\mathsf{A}%
%_{-}(\zeta _{a}^{(h_{a})}-2\varphi _{a}\xi _{a}),\text{ \ \ \ \ \ T}%
T_{a}^{\pm }|\beta,h_{1},...,h_{a},...,h_{\mathsf{N}}\rangle
=|\beta,h_{1},...,h_{a}\pm 1,...,h_{\mathsf{N}}\rangle .
\end{equation}
\end{theorem}

\begin{proof}[Proof]
The form of the action of $\mathcal{D}_{-}(\zeta _{a}^{(h_{a})}|\beta)$ on $|\beta,$%
$\mathbf{h}\rangle $ is just a consequence of the definition of the states
and the quantum determinant. Finally, the formula (\ref{R-SOV D-})  
is just a rewriting of the following interpolation formula for the
action on $|\beta$, $\mathbf{h}\rangle $:%
\begin{eqnarray*}
\mathcal{D}_{-}(\lambda |\beta)|\beta,\mathbf{h}\rangle &=&\sum_{a=1}^{2%
\mathsf{N}}T_{a}^{-\varphi _{a}}|\beta,\mathbf{h}\rangle
 \frac{\sinh(2\la-\eta)}{\sinh(2\zeta_a^{(h_a)}-\eta)}
 \prod 
_{\substack{ b=1  \\ b\neq a}}^{2\mathsf{N}}\frac{\sinh (\lambda -\zeta
_{b}^{(h_{b})})}{\sinh (\zeta _{a}^{(h_{a})}-\zeta _{b}^{(h_{b})})}f^{\varphi_a}_a(\beta)\mathsf{A}_{-}(-\zeta _{a}^{(1-h_{a})})  \notag \\
&&+|\beta,\mathbf{h}\rangle \det_{q}M(0)\cosh(\la-\eta/2)\prod_{ b=1 }^{2\mathsf{N}}\frac{\sinh (\lambda -\zeta _{b}^{(h_{b})})}{%
\sinh (\eta/2-\zeta _{b}^{(h_{b})})}  \notag \\
&&-|\beta,\mathbf{h}\rangle \coth \zeta _{-}\det_{q}M(i\pi
/2)\sinh(\la-\eta/2)\prod_{b=1}^{2\mathsf{N}}\frac{\sinh (\lambda -\zeta _{b}^{(h_{b})})}{%
\sinh (\eta/2+i\pi/2-\zeta _{b}^{(h_{b})})}.
\end{eqnarray*}
\end{proof}

%\begin{proof}

%\section{SOV-decomposition of the identity}

\subsection{Change of basis properties}

To study the properties of the SOV basis we introduce first the standard spin basis for the $2$-dimensional linear
space $\mathcal{H}_{n}$, the quantum space in the site $n$ of the chain,
\begin{equation}
\sigma _{n}^{z}|k,n\rangle =(2k-1)|k,n\rangle ,\text{ \ }k\in \{0,1\}.
\end{equation}%
Similarly, we introduce the dual $\sigma _{n}^{z}$-eigenvectors $\langle k,n|$,%
\begin{equation}
\langle k,n|\sigma _{n}^{z}=(2k-1)\langle k,n|,\text{ \ }k\in \{0,1\}.
\end{equation}
The tensor products of the local basis vectors constitute an orthogonal basis in $\mathcal{H}$
\begin{equation}
\ket{\mathbf{k}}=\otimes _{n=1}^{\mathsf{N}}|k_{n},n\rangle,\qquad \bra{\mathbf{k}}=\otimes _{n=1}^{\mathsf{N}}
\bra{k_{n},n}\quad\text{where}\quad \mathbf{k}=\{k_1,\dots,k_\mathsf{N}\} ,
\end{equation}
and
\begin{equation}
\langle \mathbf{k}'\ket{\mathbf{k}}=\prod_{n=1}^{\mathsf{N}}\delta
_{k_{n},k_{n}^{\prime }}\text{ \ \ }\forall k_{n},k_{n}^{\prime }\in
\{0,1\}.
\end{equation}
We define the following  $2^{\mathsf{N}}\times 2^{\mathsf{N}}$ matrices $U^{(L,\beta)}$
and $U^{(R,\beta)}$:%
\begin{equation}
\langle \beta,\mathbf{h}|=\langle \mathbf{h}|U^{(L,\beta)}=%
\sum_{i=1}^{2^{\mathsf{N}}}U_{\varkappa \left( \mathbf{h}\right)
,i}^{(L,\beta)}\langle \varkappa ^{-1}\left( i\right) |\text{ \ \ and\ \ \ }|\beta,%
\mathbf{h}\rangle =U^{(R,\beta)}|\mathbf{h}\rangle =\sum_{i=1}^{2^{%
\mathsf{N}}}U_{i,\varkappa \left( \mathbf{h}\right)
}^{(R,\beta)}|\varkappa ^{-1}\left( i\right) \rangle ,
\end{equation}%
which define the change of basis to the SOV-basis starting from the original
spin basis:%
\begin{equation}
\langle \mathbf{h}|= \otimes _{n=1}^{\mathsf{N}}\langle
h_{n},n|\text{ \ \ \ \ and \ \ \ }|\mathbf{h}\rangle =
\otimes _{n=1}^{\mathsf{N}}|h_{n},n\rangle ,
\end{equation}%
where $\varkappa $ is the following isomorphism between the sets $%
\{0,1\}^{\mathsf{N}}$ and $\{1,...,2^{\mathsf{N}}\}$: 
\begin{equation}
\varkappa :\mathbf{h}\in \{0,1\}^{\mathsf{N}}\rightarrow \varkappa
\left( \mathbf{h}\right) = 1+\sum_{a=1}^{\mathsf{N}%
}2^{(a-1)}h_{a}\in \{1,...,2^{\mathsf{N}}\}.  \label{corrisp}
\end{equation}%
Note that the matrices $U^{(L,\beta )}$ and $U^{(R,\beta )}$\ are
invertible matrices for the pseudo-diagonalizability of $\mathcal{B}%
_{-}(\lambda |\beta)$: 
\begin{equation}
U^{(L,\beta)}\mathcal{B}_{-}(\lambda |\beta)=\Delta _{\mathcal{B}_{-}}^{L}(\lambda
|\beta)U^{(L,\beta-2)},\text{ \ \ }\mathcal{B}_{-}(\lambda
|\beta)U^{(R,\beta)}=U^{(R,\beta+2)}\Delta _{\mathcal{B}_{-}}^{R}(\lambda |\beta).
\end{equation}%
Here $\Delta _{\mathcal{B}_{-}}^{L/R}(\lambda |\beta)$ are the $2^{\mathsf{N}%
}\times 2^{\mathsf{N}}$ diagonal matrices with elements
\begin{equation}
\left( \Delta _{\mathcal{B}_{-}}^{L}(\lambda |\beta)\right) _{i,j}= \delta
_{i,j}\mathsf{B}_{\varkappa ^{-1}\left( i\right) }(\lambda |\beta),%
\text{ }\left( \Delta _{\mathcal{B}_{-}}^{R}(\lambda |\beta)\right) _{i,j}=
\delta _{i,j}\mathsf{B}_{\varkappa ^{-1}\left( i\right)
}(\lambda |\beta),\text{ \ }\forall i,j\in \{1,...,2^{\mathsf{N}}\}.
\end{equation}%

The main result of this section is the following proposition:
\begin{proposition}
The $2^{\mathsf{N}}\times 2^{\mathsf{N}}$ matrix:%
\begin{equation}
M\equiv U^{(L,\beta-2)}U^{(R,\beta)}
\end{equation}%
is diagonal and it is characterized by:%
\begin{equation}
M_{\varkappa \left( \mathbf{h}\right) \varkappa \left( \mathbf{k%
}\right) }=\langle \beta-2,\mathbf{h}|\beta,\mathbf{k}\rangle =\delta
_{\varkappa \left( \mathbf{h}\right) \varkappa \left( \mathbf{k}%
\right) }
Z(\be-2)
\prod_{1\leq b<a\leq \mathsf{N}}\frac{1}{\eta _{a}^{(h_{a})}-\eta
_{b}^{(h_{b})}},  \label{M_jj}
\end{equation}%
with  the normalization constant%
\begin{equation}
Z(\be)= \prod_{1\leq b<a\leq \mathsf{N}}(\eta
_{a}^{(1)}-\eta _{a}^{(1)})\langle \beta|\left( \prod_{n=1}^{\mathsf{N}}%
\mathcal{A}_{-}(\eta /2-\xi _{n}|\beta+2)/\mathsf{A}_{-}(\eta /2-\xi _{n})\right) 
|-\beta\rangle ,  \label{Norm-def}
\end{equation}%
and 
\begin{equation}
\eta _{a}^{(h_{a})}\equiv \cosh 2\left[ (\xi _{a}+(h_{a}-\frac{1}{2})\eta %
\right] .
\end{equation}
\end{proposition}

\begin{proof}
First we prove  that the matrix $M$ is diagonal. %That is we want
%to prove the presence of the $\delta _{\varkappa \left( \mathbf{h}%
%\right) \varkappa \left( \mathbf{k}\right) }$ in (\ref{M_jj}); 
In order to do it we compute the matrix element $\langle \beta,$%
$\mathbf{h}|\mathcal{B}_{-}(\lambda |\beta)|\beta,\mathbf{k}\rangle $ which lead
to the following identity:%
\begin{equation}
\mathsf{B}_{\mathbf{h}}(\lambda |\beta)\langle \beta-2,\mathbf{%
h}|\beta,\mathbf{k}\rangle =\mathsf{\bar B}_{\mathbf{k}%
}(\lambda |\beta)\langle \beta,\mathbf{h}|\beta+2,\mathbf{k}\rangle ,
\end{equation}%
which implies:%
\begin{equation}
\langle \beta-2,\mathbf{h}|\beta,\mathbf{k}\rangle \propto \delta
_{\varkappa \left( \mathbf{h}\right) \varkappa \left( \mathbf{k}%
\right) },  \label{orthogonality-pseudo-states}
\end{equation}%
as from the condition $\mathbf{h}\neq \mathbf{k}$ it follows that $\exists
n\in \{1,...,\mathsf{N}\}$\ such that $h_{n}\neq k_{n}$ and then:%
\begin{equation}
\mathsf{B}_{\mathbf{h}}(\zeta _{n}^{(k_{n})}|\beta)\neq 0,\quad
\mathsf{\bar B}_{\mathbf{k}}(\zeta _{n}^{(k_{n})}|\beta)=0.
\end{equation}%
To compute the diagonal elements $M_{\varkappa \left( \mathbf{h}%
\right) \varkappa \left( \mathbf{h}\right) }$, we compute the matrix
elements 
$$\theta _{a}(\beta)= \langle
\be-2,h_{1},...,h_{a}=1,...,h_{\mathsf{N}}|\mathcal{D}_{-}(\xi _{a}+\eta
/2|\beta)|\beta,h_{1},...,h_{a}=0,...,h_{\mathsf{N}}\rangle ,$$
 where $a\in \{1,...,%
\mathsf{N}\}$. Using the right action of the operator $\mathcal{D}_{-}(\xi
_{a}+\eta /2|\beta)$ and the condition (\ref{orthogonality-pseudo-states}), we get:%
\begin{align}
\theta _{a}(\be)& = f ^{-1}_a(\beta)\mathsf{A}%
_{-}(\eta /2+\xi _{a})\frac{\sinh \eta }{\sinh (2\xi _{a}-\eta )}\prod 
_{\substack{ b=1  \\ b\neq a}}^{\mathsf{N}}\frac{\cosh 2\zeta
_{a}^{(1)}-\cosh 2\zeta _{b}^{(h_{b})}}{\cosh 2\zeta _{a}^{(0)}-\cosh 2\zeta
_{b}^{(h_{b})}}  \notag \\
& \times \left. \langle \beta-2,h_{1},...,h_{a}=1,...,h_{\mathsf{N}%
}|\beta,h_{1},...,h_{a}=1,...,h_{\mathsf{N}}\rangle \right.
\end{align}%
while using the decomposition (\ref{parity-m-3}) and the fact that:%
\begin{equation}
\langle \beta-2,h_{1},...,h_{a}=1,...,h_{\mathsf{N}}|\mathcal{A}_{-}(-(\xi
_{a}+\eta /2)|\beta)=0
\end{equation}%
it holds:%
\begin{align}
& \langle \beta-2,h_{1},...,h_{a}\left. =\right. 1,...,h_{\mathsf{N}}|\mathcal{D}%
_{-}(\xi _{a}+\eta /2|\beta)  \notag \\
& \left. =\right. \frac{\sinh \eta \sinh (2\xi _{a}+\beta \eta )}{\sinh
(2\xi _{a}+\eta )\sinh (\beta )\eta }\langle \beta-2,h_{1},...,h_{a}\left.
=\right. 1,...,h_{\mathsf{N}}|\mathcal{A}_{-}(\xi _{a}+\eta /2|\beta) \\
\left. =\right. & \frac{\sinh \eta \sinh (2\xi _{a}+\beta \eta )}{\sinh
(2\xi _{a}+\eta )\sinh (\beta )\eta }\mathsf{A}_{-}(\eta /2+\xi
_{a})\langle \beta-2,h_{1},...,h_{a}\left. =\right. 0,...,h_{\mathsf{N}}|,
\end{align}%
and then we get:%
\begin{equation}
\theta _{a}^{\left( m\right) }=\frac{\sinh \eta \sinh (2\xi _{a}+(\beta
)\eta )}{\sinh (2\xi _{a}+\eta )\sinh (\beta )\eta }\mathsf{A}_{-}(\eta
/2+\xi _{a})\langle \beta-2,h_{1},...,h_{a}=0,...,h_{\mathsf{N}%
}|\beta,h_{1},...,h_{a}=0,...,h_{\mathsf{N}}\rangle ,
\end{equation}%
so that it holds:%
\begin{equation}
\frac{\langle \beta-2,h_{1},...,h_{a}=1,...,h_{\mathsf{N}%
}|\beta,h_{1},...,h_{a}=1,...,h_{\mathsf{N}}\rangle }{\langle
m-2,h_{1},...,h_{a}=0,...,h_{\mathsf{N}}|\beta,h_{1},...,h_{a}=0,...,h_{\mathsf{N%
}}\rangle }=\prod_{\substack{ b=1  \\ b\neq a}}^{\mathsf{N}}\frac{\cosh
2\zeta _{a}^{(0)}-\cosh 2\zeta _{b}^{(h_{b})}}{\cosh 2\zeta _{a}^{(1)}-\cosh
2\zeta _{b}^{(h_{b})}},  \label{F1}
\end{equation}%
from which one can prove:%
\begin{equation}
\frac{\langle \beta-2,h_{1},...,h_{\mathsf{N}}|\beta,h_{1},...,h_{\mathsf{N}}\rangle 
}{\langle \beta-2,1,...,1|\beta,1,...,1\rangle }=\prod_{1\leq b<a\leq \mathsf{N}}%
\frac{\eta _{a}^{\left( 1\right) }-\eta _{b}^{\left( 1\right) }}{\eta
_{a}^{\left( h_{a}\right) }-\eta _{b}^{\left( h_{b}\right) }}.  \label{F2}
\end{equation}%
This prove the proposition as it is easy to see that%
\begin{equation}
\langle \beta-2,1,...,1|\beta,1,...,1\rangle =Z(\be-2)\prod_{1\leq b<a\leq \mathsf{N}}\frac{1%
}{\eta _{a}^{\left( 1\right) }-\eta _{b}^{\left( 1\right) }},
\end{equation}%
by our definition of the normalization $Z(\beta)$.
\end{proof}

\subsection{SOV-decomposition of the identity}

The identity $\mathbb{I}$ admits the following representation in terms of
left and right SOV-basis:%
\begin{equation}
\mathbb{I}= \sum_{i=1}^{2^{\mathsf{N}}}\mu |\beta,\varkappa ^{-1}\left(
i\right) \rangle \langle \beta-2,\varkappa ^{-1}\left( i\right) |,
\end{equation}%
where the $\mu = \left( \langle \beta-2,\varkappa ^{-1}\left( i\right)
|\beta,\varkappa ^{-1}\left( i\right) \rangle \right) ^{-1}$ is the Sklyanin's
measure\footnote{Sklyanin's measure has been first introduced by Sklyanin in the quantum Toda
chain \cite{Skl85}, see also \cite{Smi98a} and \cite{BytT06} for further
discussions.} analogous in our 6-vertex reflection algebra representations.
Now using the result of the previous section we can write it explicitly:%
\begin{equation}
\mathbb{I}=\frac 1{Z(\be-2)} \sum_{h_{1},...,h_{\mathsf{N}}=0}^{1}\prod_{1\leq b<a\leq 
\mathsf{N}}(\eta _{a}^{(h_{a})}-\eta _{a}^{(h_{a})})|\beta,h_{1},...,h_{\mathsf{N%
}}\rangle \langle \beta-2,h_{1},...,h_{\mathsf{N}}|.  \label{Decomp-Id}
\end{equation}

\section{SOV representations for $\mathcal{T}(\protect\lambda )$-spectral
problem}

\label{SOV-Gen}In \cite{Skl85,Skl92} Sklyanin has introduced a method to
construct quantum separation of variable (SOV) representations for the
spectral problem of the transfer matrices associated to the representations
of the Yang-Baxter algebra. For the most general representations of the
reflection algebra with non-diagonal boundary matrices the quantum SOV
representations are  constructed here following the same approach developed
in \cite{Nic12b} but we use  the gauge transformation to eliminate one of the non-diagonal entries of $K_+$. It means that we fix either $\al-\beta$ or $\al+\beta$. It is important to underline that the second gauge parameter remains free and can be used either to eliminate the second non-diagonal entry of $K_+$ or the corresponding entry of $K_-$. However we do not need to fix this second parameter to construct the eigenvectors of the transfer matrix. 

More precisely,
 the following theorems hold:

\begin{theorem}
\label{bsovtau}
Under the most general boundary conditions,
and  if the gauge parameters $\alpha ,\beta \in \mathbb{C}$ satisfy the following condition for an integer $k$
\begin{equation}
(\alpha -\beta +2)\eta = -\tau _{+}+(-1)^k (\alpha _{+}-\beta _{+})+i\pi
k,  \label{Triangular-gauge-K+B}
\end{equation}%
then $K_+^{(L)}(\la|\beta-1)_{12}=K_+^{(R)}(\la|\beta-1)_{12}=0$ and:

I$_{b}$) the left representation for which the one parameter family $%
\mathcal{B}_{-}(\lambda |\beta-2)$\ is pseudo-diagonal defines a left SOV
representation for the spectral problem of the transfer matrix $\mathcal{T}%
(\lambda )$.

II$_{b}$) the right representation for which the one parameter family $%
\mathcal{B}_{-}(\lambda |\beta)$\ is pseudo-diagonal defines a right SOV
representation for the spectral problem of the transfer matrix $\mathcal{T}%
(\lambda )$.
\end{theorem}

Similarly we can formulate the same theorem for the $\mathcal{C}_{-}(\lambda |\beta)$ SOV representations: 

\begin{theorem}
\label{csovtau}
Under the most general boundary conditions, if the gauge parameters $\alpha ,\beta \in \mathbb{C}$ satisfy the following condition for an integer $k$
\begin{equation}
(\alpha +\beta )\eta = -\tau _{+}+(-1)^k (\alpha _{+}-\beta _{+})+i\pi
k,  \label{Triangular-gauge-K+C}
\end{equation}%
then $K_+^{(L)}(\la|\beta-1)_{21}=K_+^{(R)}(\la|\beta-1)_{21}=0$ and:

I$_{c}$) the left representation for which the one parameter family $%
\mathcal{C}_{-}(\lambda |\beta+2)$\ is pseudo-diagonal defines a left SOV
representation for the spectral problem of the transfer matrix $\mathcal{T}%
(\lambda )$.

II$_{c}$) the right representation for which the one parameter family $%
\mathcal{C}_{-}(\lambda |\beta)$\ is pseudo-diagonal defines a right SOV
representation for the spectral problem of the transfer matrix $\mathcal{T}%
(\lambda )$.
\end{theorem}

The proof of the Theorem \ref{bsovtau} and the explicit constructions of the SOV solutions
of the spectral problem for the transfer matrix $\mathcal{T}(\lambda )$ will
be given in the following subsections. Theorem \ref{csovtau}  can be proved  in a similar
way.

\subsection{Transfer matrix spectrum in $\mathcal{B}_{-}(\la|\beta)$%
-SOV-representations}

\begin{theorem}
\label{C:T-eigenstates-} Let  $\Sigma _{\mathcal{T}}$ be  the set
of the eigenvalue functions of the transfer matrix $\mathcal{T}(\lambda )$,
then any $\tau (\lambda )\in \Sigma _{\mathcal{T}}$ is an even function of $%
\lambda $\ of the form:%
\begin{align}
\tau (\lambda )=&\sum_{a=1}^{\mathsf{N}}\frac{\cosh^2 2\la-\cosh^2\eta}{\cosh^2 2\zeta
_{a}^{(0)}-\cosh^2\eta}\prod_{\substack{ b=1  \\ b\neq a}}%
^{\mathsf{N}}\frac{\cosh 2\lambda -\cosh 2\zeta _{b}^{(0)}}{\cosh 2\zeta
_{a}^{(0)}-\cosh 2\zeta _{b}^{(0)}}\tau (\zeta _{a}^{(0)}), \notag\\
&+(\cosh 2\la+\cosh\eta)\prod_{b=1}%
^{\mathsf{N}}\frac{\cosh 2\lambda -\cosh 2\zeta _{b}^{(0)}}{\cosh \eta
-\cosh 2\zeta _{b}^{(0)}} \det_{q}M(0) \notag\\
&+(-1)^{\mathsf{N}}(\cosh 2\la-\cosh\eta)\prod_{b=1}%
^{\mathsf{N}}\frac{\cosh 2\lambda -\cosh 2\zeta _{b}^{(0)}}{\cosh \eta
+\cosh 2\zeta _{b}^{(0)}} \coth\zeta_-\coth\zeta_+\det_{q}M(i\pi/2).
 \label{set-T-1}
\end{align}%
%where: 
%\begin{equation}
%\tau (\pm \zeta _{\mathsf{N}+1}^{(0)})=2\cosh \eta \det_{q}M(0),\text{ \ }%
%\tau (\pm \zeta _{\mathsf{N}+1}^{(0)})=-2\cosh \eta \coth \zeta _{-}\coth
%\zeta _{+}\det_{q}M(i\pi /2).  \label{set-T-2}
%\end{equation}%
If the condition (\ref{E-SOV}) is satisfied, then $\mathcal{T}%
(\lambda )$ has simple spectrum and $\Sigma _{\mathcal{T}}$ is given by 
the solutions of the discrete system of equations: 
\begin{equation}
\tau (\pm \zeta _{a}^{(0)})\tau (\pm \zeta _{a}^{(1)})=\mathbf{A}%
(\zeta _{a}^{(1)})\mathbf{A}(-\zeta _{a}^{(0)}),\quad\forall
a\in \{1,...,\mathsf{N}\},  \label{I-Functional-eq}
\end{equation}%
in the class of functions of the form (\ref{set-T-1}),
where the coefficient $\mathbf{A}(\lambda )$ is defined by: 
\begin{equation}
\mathbf{A}(\lambda )\equiv \mathsf{a}_{+}(\lambda|\beta -1)\mathsf{A}%
_{-}(\lambda ),
\end{equation}%
and satisfies the quantum determinant condition:

\begin{equation}
\frac{\det_{q}K_{+}(\lambda)\det_{q}
\mathcal{U}_{-}(\lambda)}{\sinh (2\lambda +\eta )\sinh (2\lambda -\eta )}=\mathbf{A}(\lambda+\eta/2 )\mathbf{A}(-\lambda +\eta/2 ).\label{Tot-q-det-tt}
\end{equation}

\begin{itemize}
\item[\textsf{I)}] Under the condition (\ref{NON-nilp-B-R}),
the vector:%
\begin{equation}
|\tau \rangle =\sum_{h_{1},...,h_{\mathsf{N}}=0}^{1}\prod_{a=1}^{\mathsf{N}%
}Q_{\tau }(\zeta _{a}^{(h_{a})})\prod_{1\leq b<a\leq \mathsf{N}}(\eta
_{a}^{(h_{a})}-\eta _{b}^{(h_{b})})|\beta,h_{1},...,h_{\mathsf{N}}\rangle ,
\label{eigenT-r-D}
\end{equation}%
defines, uniquely up to an overall normalization, the right $\mathcal{T}$%
-eigenstate corresponding to $\tau (\lambda )\in \Sigma _{\mathcal{T}}$. The
coefficients in (\ref{eigenT-r-D}) are characterized by:%
\begin{equation}
Q_{\tau }(\zeta _{a}^{(1)})/Q_{\tau }(\zeta _{a}^{(0)})=\tau (\zeta
_{a}^{(0)})\mathbf{A}(-\zeta _{a}^{(0)}).  \label{t-Q-relation}
\end{equation}

\item[\textsf{II)}] Under the condition  (\ref{NON-nilp-B-L}),
the covector 
\begin{equation}
\langle \tau |=\sum_{h_{1},...,h_{\mathsf{N}}=0}^{1}\prod_{a=1}^{\mathsf{N}}%
\bar{Q}_{\tau }(\zeta _{a}^{(h_{a})})\prod_{1\leq b<a\leq \mathsf{N}}(\eta
_{a}^{(h_{a})}-\eta _{b}^{(h_{b})})\langle \beta-2,h_{1},...,h_{\mathsf{N}}|,
\label{eigenT-l-D}
\end{equation}%
defines, uniquely up to an overall normalization, the left $\mathcal{T}$%
-eigenstate corresponding to $\tau (\lambda )\in \Sigma _{\mathcal{T}}$. The
coefficients in (\ref{eigenT-l-D}) are characterized by:%
\begin{equation}
\bar{Q}_{\tau _{-}}(\zeta _{a}^{(1)})/\bar{Q}_{\tau _{-}}(\zeta
_{a}^{(0)})=\tau (\zeta _{a}^{(0)})/\mathbf{D}(\zeta _{a}^{(1)}),
\label{t-Qbar-relation}
\end{equation}%
where:%
\begin{equation}
\mathbf{D}(\zeta _{a}^{(h_a)} )\equiv \mathsf{d}_{+}(\zeta _{a}^{(h_a)} |\beta-1 )\mathsf{D}%
_{-}(\zeta _{a}^{(h_a)}  ).  \label{coeff-T-D}
\end{equation}
\end{itemize}
\end{theorem}

\begin{proof}
The transfer matrix $\mathcal{T}(\lambda )$\ is an even function of $\lambda 
$\ so the same is true for the $\tau (\lambda )\in \Sigma _{\mathcal{T}}$ .
Moreover, from the identities (\ref{U-identities}) and after some simple
computation the following identities are derived: 
\begin{align}
\mathcal{T}(\pm \eta /2)& =2\cosh \eta \det_{q}M(0), \\
\mathcal{T}(\pm (\eta /2-i\pi /2))& =-2\cosh \eta \coth \zeta _{-}\coth
\zeta _{+}\det_{q}M(i\pi /2).
\end{align}%
These identities together with the known functional form of $\mathcal{T}%
(\lambda )$ with respect to  $\lambda $ imply that $\tau (\lambda )\in \Sigma _{%
\mathcal{T}}$ satisfy the characterization (\ref{set-T-1}).
 In the $\mathcal{B}_{-}$-SOV representations the spectral problem for $%
\mathcal{T}(\lambda )$ is reduced to the following discrete system of $2^{%
\mathsf{N}}$ Baxter-like equations: 
\begin{equation}
\tau (\zeta _{n}^{(h_{n})})\Psi _{\tau }(\mathbf{h})\,=\mathbf{A%
}(\zeta _{n}^{(h_{n})})\Psi _{\tau }(\mathsf{T}_{n}^{-}(\mathbf{h}))+%
\mathbf{A}(-\zeta _{n}^{(h_{n})})\Psi _{\tau }(\mathsf{T}_{n}^{+}(%
\mathbf{h})),  \label{SOVBax1}
\end{equation}%
for any $n\in \{1,...,\mathsf{N}\}$ and $\mathbf{h}\,\in \{0,1\}^{\mathsf{N}%
} $, in the coefficients (\textit{wave-functions}) $\Psi _{\tau }(\mathbf{h}%
) $ of the $\mathcal{T}$-eigenstate $|\tau \rangle $ associated to $\tau
(\lambda )\in \Sigma _{\mathcal{T}}$. Here, we have used the notations:%
\begin{equation}
\mathsf{T}_{n}^{\pm }(\mathbf{h})= (h_{1},\dots ,h_{n}\pm
1,\dots ,h_{\mathsf{N}}).
\end{equation}%
This system trivially follows when we recall the identities:%
\begin{equation}
\mathbf{A}_{-}(\zeta _{n}^{(0)})=\mathbf{A}_{-}(-\zeta
_{n}^{(1)})=0,
\end{equation}%
and we compute the matrix elements:%
\begin{equation}
\langle \beta-2,h_{1},...,h_{n},...,h_{\mathsf{N}}|\mathcal{T}(\pm \zeta
_{n}^{(h_{n})})|\tau \rangle .
\end{equation}%
Indeed, from the decomposition (\ref{T-decomp-L}), we have:%
\begin{align}
\tau (\pm \zeta _{n}^{(0)})\Psi _{\tau }(h_{1},...,h_{n}& =0,...,h_{\mathsf{N%
}})=  \notag \\
& =\langle \beta-2,h_{1},...,h_{n}=0,...,h_{\mathsf{N}}|\mathcal{T}(-\zeta
_{n}^{(0)})|\tau \rangle  \notag \\
& =\mathsf{a}_{+}(-\zeta _{n}^{(0)})\langle \beta-2,h_{1},...,h_{n}=0,...,h_{%
\mathsf{N}}|\mathcal{A}_{-}(-\zeta _{n}^{(0)})|\tau \rangle  \notag \\
& =\mathbf{A}(-\zeta _{n}^{(0)})\Psi _{\tau
}(h_{1},...,h_{n}=1,...,h_{\mathsf{N}}),
\end{align}%
and%
\begin{align}
\tau (\pm \zeta _{n}^{(1)})\Psi _{\tau }(h_{1},...,h_{n}& =1,...,h_{\mathsf{N%
}})=  \notag \\
& =\langle \beta-2,h_{1},...,h_{n}=1,...,h_{\mathsf{N}}|\mathcal{T}(\zeta
_{n}^{(1)})|\tau \rangle  \notag \\
& =\mathsf{a}_{+}(\zeta _{n}^{(1)})\langle \beta-2,h_{1},...,h_{n}=1,...,h_{%
\mathsf{N}}|\mathcal{A}_{-}(\zeta _{n}^{(1)})|\tau \rangle  \notag \\
& =\mathbf{A}(\zeta _{n}^{(1)})\Psi _{\tau }(h_{1},...,h_{n}=0,...,h_{%
\mathsf{N}}) .
\end{align}

Clearly the previous system of equations (\ref{SOVBax1})  is
equivalent to the following system of homogeneous equations:%
\begin{equation}
\left( 
\begin{array}{cc}
\tau (\pm \zeta _{n}^{(0)}) & -\mathbf{A}(-\zeta _{n}^{(0)}) \\ 
-\mathbf{A}(\zeta _{n}^{(1)}) & \tau (\pm \zeta _{n}^{(1)})%
\end{array}%
\right) \left( 
\begin{array}{c}
\Psi _{\tau -}(h_{1},...,h_{n}=0,...,h_{1}) \\ 
\Psi _{\tau -}(h_{1},...,h_{n}=1,...,h_{1})%
\end{array}%
\right) =\left( 
\begin{array}{c}
0 \\ 
0%
\end{array}%
\right) ,  \label{homo-system}
\end{equation}%
for any $n\in \{1,...,\mathsf{N}\}$ with $h_{m\neq n}\in \{0,1\}$. The
condition $\tau (\lambda )\in \Sigma _{\mathcal{T}_{-}}$ implies that the
determinants of the $2\times 2$ matrices in (\ref{homo-system}) 
 must be zero for any $n\in \{1,...,\mathsf{N}\}$, which is equivalent to (%
\ref{I-Functional-eq}). Moreover, the rank of the matrices in (\ref%
{homo-system})  is 1 as%
\begin{equation}
\mathbf{A}(-\zeta _{n}^{(0)})\neq 0\text{\ \ and \ \ }\mathbf{A}(\zeta
_{n}^{(1)})\neq 0,  \label{Rank1}
\end{equation}%
and then (up to an overall normalization) the solution is unique:%
\begin{equation}
\frac{\Psi _{\tau }(h_{1},...,h_{n}=1,...,h_{\mathsf{N}})}{\Psi _{\tau
}(h_{1},...,h_{n}=0,...,h_{\mathsf{N}})}=\frac{\tau (\zeta _{a}^{(0)})}{%
\mathbf{A}(-\zeta _{a}^{(0)})},
\end{equation}%
for any $n\in \{1,...,\mathsf{N}\}$ with $h_{m\neq n}\in \{0,1\}$. So fixed $%
\tau (\lambda )\in \Sigma _{\mathcal{T}}$ there exists (up to normalization)
one and only one corresponding $\mathcal{T}$-eigenstate $|\tau \rangle $
with coefficients of the factorized form given in (\ref{eigenT-r-D})-(\ref{t-Q-relation}); i.e. the $\mathcal{T}$%
-spectrum is simple.

Vice versa, if $\tau (\lambda )$ is in the set of functions (\ref{set-T-1}) %
 and satisfies (\ref{I-Functional-eq}), then the state $%
|\tau \rangle $ defined by (\ref{eigenT-r-D}-\ref%
{t-Q-relation})  satisfies:%
\begin{eqnarray*}
\left\langle \beta-2,h_{1},...,h_{n},...,h_{\mathsf{N}}\right\vert \mathcal{T}%
(\zeta _{n}^{(h_{n})})|\tau \rangle &=&\left\{ \!\!
\begin{array}{l}
\mathbf{A}(-\zeta _{n}^{(0)})\Psi _{\tau }(h_{1},...,h_{n}=1,...,h_{%
\mathsf{N}})\text{ \ \ for }h_{n}=0 \\ 
\mathbf{A}(\zeta _{n}^{(1)})\Psi _{\tau }(h_{1},...,h_{n}=0,...,h_{%
\mathsf{N}})\text{ \ \ for }h_{n}=1%
\end{array}%
\right. \\
&=&\left\{ \!\!
\begin{array}{l}
\mathbf{A}(-\zeta _{n}^{(0)})\frac{\tau (\zeta _{a}^{(0)})}{%
\mathbf{A}(-\zeta _{a}^{(0)})}\Psi _{\tau
}(h_{1},...,h_{n}=0,...,h_{\mathsf{N}})\text{ \  for }h_{n}=0 \\ 
\mathbf{A}(\zeta _{n}^{(1)})\frac{\tau (\zeta _{a}^{(0)})}{%
\mathbf{A}(\zeta _{a}^{(1)})}\Psi _{\tau
}(h_{1},...,h_{n}=1,...,h_{\mathsf{N}})\text{ \  for }h_{n}=1%
\end{array}%
\right. \\
&=&\tau (\zeta _{n}^{(h_{n})})\Psi _{\tau }(h_{1},...,h_{n},...,h_{\mathsf{N}%
})\text{ \ }\forall n\in \{1,...,\mathsf{N}\},
\end{eqnarray*}%
this, and the following functional form with respect to  $\lambda $
of the transfer matrix:%
\begin{equation}
\mathcal{T}(\lambda )=\sum_{b=1}^{\mathsf{N}+2}\mathcal{T}_{b}(\cosh
2\lambda )^{b-1},  \label{set-t}
\end{equation}%
implies the identity:%
\begin{equation}
\left\langle \beta-2,h_{1},...,h_{\mathsf{N}}\right\vert \mathcal{T}(\lambda
)|\tau \rangle =\tau (\lambda )\Psi _{\tau }(h_{1},...,h_{n},...,h_{\mathsf{N%
}})\text{ \ \ }\forall \lambda \in \mathbb{C},
\end{equation}%
for any $\mathcal{B}_{-}(\la|\beta-2)$ pseudo-eigenstate $\left\langle \beta-2,h_{1},...,h_{%
\mathsf{N}}\right\vert $, i.e. $\tau (\lambda )\in \Sigma _{\mathcal{T}}$\,  
and $|\tau \rangle $ is the corresponding eigenstate of the transfer matrix $\mathcal{T}$. The proof  for  the left $\mathcal{T}$-eigenstates is very similar and we skip it here.

Finally, it is important to point out that the quantum determinant condition \rf{Tot-q-det-tt} is a simple consequence of the following identity
\begin{equation}
\det_{q}K_{+}(\lambda)=-\sinh (2\lambda+\eta )\mathsf{a}_{+}(\lambda+\eta/2|\beta-1 )\mathsf{a}_{+}(-\lambda +\eta/2|\be-1 )
\end{equation}
which can be proven by direct computations when the condition \rf{Triangular-gauge-K+B} is satisfied.
\end{proof}

This theorem implies that each eigenvalue and eigenstate of the transfer matrix can be characterized in terms of a set of parameters  $\{x_{1},...,x_{\mathsf{N}}\}$ satisfying a system of quadratic equations. This system replaces the Bethe equations in this case. More precisely:  
\begin{corollary}
The set $\Sigma _{\mathcal{T}}$ of the eigenvalue functions of the
transfer matrix $\mathcal{T}(\lambda )$ admits the following
characterization:%
\begin{equation}
\Sigma _{\mathcal{T}}= \left\{ \tau (\lambda ):\tau (\lambda
)=f(\lambda )+\sum_{a=1}^{\mathsf{N}}g_{a}(\lambda )x_{a},\text{ \ \ }%
\forall \{x_{1},...,x_{\mathsf{N}}\}\in \Sigma _{T}\right\} ,
\end{equation}%
where we have defined:%
\begin{equation}
g_{a}(\lambda )= \frac{\cosh^2 2\la-\cosh^2\eta}{\cosh^2 2\zeta
_{a}^{(0)}-\cosh^2\eta}\,\prod_{\substack{ b=1  \\ %
b\neq a}}^{\mathsf{N}}\frac{\cosh 2\lambda -\cosh 2\zeta _{b}^{(0)}}{\cosh
2\zeta _{a}^{(0)}-\cosh 2\zeta _{b}^{(0)}}\quad\text{ \ for }a\in \{1,...,\mathsf{%
N}\},
\end{equation}
\begin{align}
f(\lambda )=& \frac{(\cosh 2\la+\cosh\eta)}{2\cosh\eta}\prod_{b=1}%
^{\mathsf{N}}\frac{\cosh 2\lambda -\cosh 2\zeta _{b}^{(0)}}{\cosh \eta
-\cosh 2\zeta _{b}^{(0)}}\tau(\eta/2)  \notag\\
&-(-1)^{\mathsf{N}}\frac{(\cosh 2\la-\cosh\eta)}{2\cosh\eta}\prod_{b=1}%
^{\mathsf{N}}\frac{\cosh 2\lambda -\cosh 2\zeta _{b}^{(0)}}{\cosh \eta
+\cosh 2\zeta _{b}^{(0)}} \tau
(\eta/2+i\pi/2),
\end{align}%
and $\Sigma _{T}$ is the set of the solutions to the following inhomogeneous system of $%
\mathsf{N}$ quadratic equations:%
\begin{equation}
x_{n}\sum_{a=1}^{\mathsf{N}}g_{a}(\zeta _{n}^{(1)})x_{a}+x_{n}f(\zeta
_{n}^{(1)})=q_{n},\text{ \ \ \ }q_{n}= \frac{\det_{q}K_{+}(\xi_n)\det_{q}
\mathcal{U}_{-}(\xi_n)}{\sinh(\eta+2\xi_n)\sinh (\eta-2\xi_n)},\text{ \ \ }\forall n\in \{1,...,%
\mathsf{N}\},  \label{SOV-by-sys-2-eq}
\end{equation}%
in $\mathsf{N}$ parameters $\{x_{1},...,x_{\mathsf{N}}\}$.
\end{corollary}

\subsection{SOV applicability and Nepomechie's constraint}

Combining together conditions for the existence of SOV basis (\ref{NON-nilp-B-L}-\ref{NON-nilp-C-R}) and the choice of the gauge parameters necessary to construct the eigenstates of the transfer matrix (\ref{Triangular-gauge-K+B}- \ref{Triangular-gauge-K+C}) we obtain the limits of applicability of the SOV method. It happens to be related to the constrain situation where algebraic Bethe ansatz works, more precisely the following theorem holds.

\begin{theorem}
The SOV constructions corresponding to the cases I$_{b}$ and I$_{c}$ fails
to exist if and only if the following condition on the parameters of the
boundary matrices are
satisfied%
\begin{equation}
\left( \mathsf{N}-1\right) \eta =\tau _{-}-\tau _{+}+(-1)^k(\alpha _{-}+\beta _{-})-(-1)^m(\alpha _{+}-\beta
_{+})+i\pi (k+m), \label{Fail-SOV-i}
\end{equation}%
 where $k$ and $m$ are arbitrary integers. Similarly, the SOV constructions corresponding to the cases II$_{b}$ and II$_{c}$ fails to exist if and only if the following condition on
the parameters of the boundary matrices are
satisfied%
\begin{equation}
\left( 1-\mathsf{N}\right) \eta =\tau _{-}-\tau _{+}+(-1)^k(\alpha _{-}+\beta _{-})-(-1)^m(\alpha _{+}-\beta
_{+})+i\pi (k+m).  \label{Fail-SOV-ii}
\end{equation}%
  Then our SOV schema to construct the spectrum (eigenvalues and
eigenstates) of the transfer matrix $\mathcal{T}(\lambda )$ cannot be used if
and only if the conditions (\ref{Fail-SOV-i}) and (\ref%
{Fail-SOV-ii})  are simultaneously satisfied.
\end{theorem}

\textbf{Remark }In our notations for the boundary
parameters the Nepomechie's constraint reads%
\begin{equation}
\mathsf{k}\eta =\tau _{-}-\tau _{+}+\epsilon _{-}(\alpha _{-}+\beta
_{-})+\epsilon _{+}(\alpha _{+}-\beta _{+}),\text{ }\text{mod}\,2\pi i\text{
and }\mathsf{k}=\mathsf{N}-1+2r\text{ \ with }r\in \mathbb{Z}
\end{equation}%
so that we recover the relations (\ref{Fail-SOV-i})  and %
(\ref{Fail-SOV-ii})  respectively for $r=0$ and $r=1-\mathsf{N}$%
. The previous theorem says that the SOV construction works also when the
boundary parameters satisfy one Nepomechie's condition: if $r\neq 0$ and $%
r\neq 1-\mathsf{N}$ we can use both the left and right SOV construction, if $%
r=0$ we can use the right SOV construction and if $r=1-\mathsf{N}$ we can
use the left SOV construction. The only problem in our SOV schema appears if
the two Nepomechie's conditions for $r=0$ and $r=1-\mathsf{N}$ are
simultaneously satisfied.

Finally,  the  special case when only one of these two conditions is satisfied maybe of particular interest as in this situation there are  two simultaneous descriptions and it is possible to compare  the construction of eigenvalues and eigenstates by the separation of variables and by the algebraic Bethe ansatz. 

\section{Scalar Products}
One of the main reasons of interest in the SOV method is that it seems to provide a possibility to go beyond the spectral analysis constructing dynamic observables of the physical system. The following theorem represents the first step in the solution of this  problem. 
\begin{proposition}
Let $\langle \omega |$ and $|\rho \rangle $ be an arbitrary covector and
vector of separate forms: 
\begin{align}
\langle \omega |& =\sum_{h_{1},...,h_{\mathsf{N}}=0}^{1}\prod_{a=1}^{\mathsf{%
N}}\omega _{a}(\zeta _{a}^{(h_{a})})\prod_{1\leq b<a\leq \mathsf{N}}(\eta
_{a}^{(h_{a})}-\eta _{b}^{(h_{b})})\langle \beta,h_{1},...,h_{\mathsf{N}}|,
\label{Fact-left-SOV} \\
|\rho \rangle & =\sum_{h_{1},...,h_{\mathsf{N}}=0}^{1}\prod_{a=1}^{\mathsf{N%
}}\rho _{a}(\zeta _{a}^{(h_{a})})\prod_{1\leq b<a\leq \mathsf{N}}(\eta
_{a}^{(h_{a})}-\eta _{b}^{(h_{b})})|\beta+2,h_{1},...,h_{\mathsf{N}}\rangle ,
\label{Fact-right-SOV}
\end{align}%
in the $\mathcal{B}$-pseudo-eigenbasis, then the action of $\langle \omega |$
on $|\rho \rangle $ reads:%
\begin{equation}
\langle \omega |\rho \rangle =Z(\beta-2)\det_{\mathsf{N}}||\mathcal{M}_{a,b}^{\left(
\omega ,\rho \right) }||\text{ \ \ with \ }\mathcal{M}_{a,b}^{\left( \omega
,\rho \right) }= \sum_{h=0}^{1}\omega _{a}(\zeta _{a}^{(h)})\rho
_{a}(\zeta _{a}^{(h)})(\eta _{a}^{(h)})^{(b-1)}.  \label{Scalar-p1}
\end{equation}%
The above formula holds, in particular, if the left and right states are
transfer matrix eigenstates.
\end{proposition}

\begin{proof}
The formula (\ref{M_jj})  and the SOV-decomposition of the
states $\langle \omega |$ and $|\rho \rangle $ implies:%
\begin{equation}
\langle \omega |\rho \rangle =Z(\be-2)\sum_{h_{1},...,h_{\mathsf{N}}=0}^{1}V(\eta
_{1}^{(h_{1})},...,\eta _{\mathsf{N}}^{(h_{\mathsf{N}})})\prod_{a=1}^{%
\mathsf{N}}\omega _{a}(\zeta _{a}^{(h_{a})})\rho _{a}(\zeta _{a}^{(h_{a})}),
\end{equation}%
where 
$$V(x_{1},...,x_{\mathsf{N}})\equiv \prod_{1\leq b<a\leq \mathsf{N}%
}(x_{a}-x_{b})$$ 
is the Vandermonde determinant which due to the multilinearity
of the determinant implies (\ref{Scalar-p1}).
\end{proof}
The normalization coefficient $Z(\beta-2)$ is an artifact of the gauge transformation, for any interesting quantity (form-factors, correlation functions) represented as a ratio of two scalar products this constant will disappear.

\section*{Conclusion and outlook}
We have shown in this paper that the separation of variables can be applied to construct the eigenstates of the quantum spin chains with the most general boundary terms. These states are characterized by the roots of a system of $\mathsf{N}$ quadratic equations which replaces the Bethe equations in this general case. We also compute scalar products (up to an unphysical normalization constant). This representation provides a possibility to compute explicitly form factors and correlation functions.

Furthermore a very similar SOV analysis can be developed for the spectral  problem of transfer matrices associated to representations of the 8-vertex reflection algebra \cite{FaldN13} corresponding to the most general open XYZ spin chains.

After this paper was completed we became aware of  the recent and interesting results reported in \cite{CaoYSW13-4}. The authors construct the $T$-$Q$ functional equations for the spin chains with non-diagonal 
boundaries\footnote{See also the recent series of papers \cite{CaoYSW13-1,CaoYSW13-2,CaoYSW13-3} for the application of the same method to different models.}  and thus they obtain the transfer matrix eigenvalues.
 An important achievement of  \cite{CaoYSW13-4} is that the equation of
type \rf{I-Functional-eq} are associated to a system of Bethe equations leading to a more traditional analysis of the eigenvalue problem. 
 It would be interesting to establish a connection between our SOV construction and this approach, in particular, with the new generalized $T$-$Q$ relation.

\section*{Acknowledgments}
The authors would like to thank J.M. Maillet, V. Terras and K. Kozlowski for discussions. G.N. is supported by National Science
Foundation grants PHY-0969739 and is grateful to the YITP Institute
of Stony Brook, where he had the opportunity to develop his research
programs and the privilege to have stimulating discussions with B. M. McCoy.
N.K, is supported by ANR grant  ANR-10-BLAN-0120-04-DIADEMS.
S.F. is supported by the Burgundy region.
G.N. would like to thank the Mathematical Physics Group at IMB of the Dijon
University for their hospitality.
G.N and N.K  would like to thank the Theoretical Physics Group of the Laboratory of
Physics at ENS-Lyon for hospitality. N.K. is grateful to the LPTHE laboratory (University Paris 6) for hospitality.

\appendix
\section{Gauge transformed boundary matrices}

We give here the explicit form of the gauge transformed boundary matrices $K_{+}^{(L)}(\lambda |\beta)$ and $K_{+}^{(R)}(\lambda |\beta)$
\begin{align}
K_{+}^{(L)}(\lambda |&\beta)_{11} = \frac {1}{\sinh \beta \eta \sinh
\zeta _{+}}\Big[ \sinh \zeta _{+}\cosh (\lambda +\eta /2)\sinh
(\lambda -\eta /2+\beta \eta)   \notag \\
&  -\left( \cosh \zeta _{+}\sinh (\lambda +\eta /2)\cosh (\lambda
-\eta /2+\beta \eta) +\kappa _{+}\sinh (2\lambda +\eta )\sinh (\tau
_{+}+(\alpha +2)\eta) \right) \Big] \\
K_{+}^{(L)}(\lambda |&\beta)_{12} = \frac{e^{(\beta +1)\eta }\sinh
(2\lambda +\eta )\left[ \kappa _{+}\sinh ((\beta -1-\alpha )\eta-\tau _{+})
-e^{-\zeta _{+}}/2\right] }{\sinh \beta \eta \sinh \zeta _{+}} \\
K_{+}^{(L)}(\lambda |&\beta)_{21} = \frac{e^{-(\beta-1 )\eta }\sinh
(2\lambda +\eta )\left[ \kappa _{+}\sinh ((\beta +\alpha +1)\eta+\tau _{+})
+e^{-\zeta _{+}}/2\right] }{\sinh \beta \eta \sinh \zeta _{+}} \\
K_{+}^{(L)}(\lambda |&\beta)_{22} = \frac{1}{ \sinh (\beta +1)\eta \sinh
\zeta _{+}}\Big[ \sinh \zeta _{+}\cosh (\lambda +\eta /2)\sinh
(-\lambda +\eta /2+\beta \eta)   \notag \\
& -\left( \cosh \zeta _{+}\sinh (\lambda +\eta /2)\cosh (-\lambda
+\eta /2+\beta \eta) +\kappa _{+}\sinh (2\lambda +\eta )\sinh ((\alpha +2)\eta+\tau
_{+}) \right) \Big]
\end{align}%
and%
\begin{align}
K_{+}^{(R)}(\lambda |\beta)_{11}& = \frac{ e^{\zeta _{+}}\sinh
(\beta -1)\eta -e^{-\zeta _{+}}\sinh (2\lambda +\beta\eta )-2\kappa
_{+}\sinh (2\lambda +\eta )\sinh (\tau _{+}+\alpha \eta )}{2\sinh
\beta \eta \sinh \zeta _{+}} \\
K_{+}^{(R)}(\lambda |\beta)_{12}& = e^{-2\eta }K_{+}^{(L)}(\lambda |\beta)_{12},%
\text{ \ }K_{+}^{(R)}(\lambda |\beta)_{21}= e^{-2\eta }K_{+}^{(L)}(\lambda
|\beta)_{21} \\
K_{+}^{(R)}(\lambda |\beta)_{22}& = \frac{ e^{-\zeta _{+}}\sinh
(2\lambda -\beta\eta )+e^{\zeta _{+}}\sinh (\beta +1)\eta +2\kappa
_{+}\sinh (2\lambda +\eta )\sinh (\tau _{+}+\alpha \eta)  }{2\sinh
\beta \eta \sinh \zeta _{+}}.
\end{align}

\end{document}